\documentclass[11pt,a4paper]{article}
\usepackage{graphicx,color,amsmath,amssymb,amsthm,mathabx,bm,booktabs,cancel}
\usepackage{longtable}

\author{Independent Researcher \\[2mm] Yong Tan}

\date{}
\title{Construct Graph Logic}
\theoremstyle{plain}
\newtheorem{theorem}{Theorem}[subsection]
\newtheorem{definition}{Definition}[subsection]
\newtheorem{descussion}{Discussion}[subsection]
\newtheorem{lemma}{Lemma}[subsection]

\begin{document}
\maketitle
\begin{abstract}
In this paper, author uses set theory to construct a logic model of abstract figure from binary relation. Based on the uniform quantified structure, author gives two logic system for graph traversal and graph coloring respectively, moreover shows a new method of cutting graph. Around this model, there are six algorithms in this paper including exact graph traversal, Algebra calculation of natural number, graph partition and graph coloring. 
\end{abstract}
\section{Introduction}
\textbf{Background. }The graph theory indeed presents relevance with abstract relation among objects. Since author wrote the BOTS algorithm for graph traversal at April, 2012, author still studies this problem. We find the fact that an instance can be computed by BOTS algorithm without concern the graph classes. We guess that there is a logic model support this phenomenon, which leads to the approach possess the capacity of data-oriented. The strategy of equivalent visiting is posed so that we can quantify the process of graph traversal with monotone decreasing function. The uniform data structure of graph is set up with the method of partition of set, and this core idea may be used to solve other problems of graph, such as graph partition and graph coloring. It makes these problems may present a quantified model for the abstract relation. And this abstract logic model can guarantee algorithms possess much more general and stronger.\\
~\newline
\textbf{Related }Work. We formally state the binary relation on graph. This relation contribution can be underlying basis for this model. We repeatedly abstract the basic relation for our new logical mode, and obtain new classes or new properties of relation in new model. These algorithms we given always may have no associated with weights to each edge or arc. It makes the new relation is easy to present a practical instance. The theoretical proof and computing of relation are transformed to Algebra of sets. Furthermore, these algorithms present more reliable, intuitive, simple and high precision, although they are heuristic approaches. First, we introduce BOST and OBOTS algorithms, which runtime complexity both are $O(mn^2B)$. They can exactly compute any connected graph classes, including difficult mixed graph. All connected instances may be explored by these algorithms without recursive method as \emph{Dynamic programming}\cite{1}, such that greatly reduce the complexity of program. We can really and easily achieve the aim of parallel and distributed computing for graph traversal.%

Graph partition may be independent of weight not like Kernighan Lin algorithm\cite{2} 
, although its runs on time complexity $O(n^3)$. In this thesis, the graph partition actually is a method of cutting graph. You can arbitrarily choose the nodes on instance for your research of AI, network flow, graph color, physical problems and etc. It makes the abstract relation among nodes be partition on a sequence of domains for your model of problem.%

Graph coloring is not a simple labeling each vertex on instance. It becomes a logical problem for how to cut graph and let those vertices be partitioned to two classes. Author gives two speed-up algorithms BOGPC and BOERC, which can run in $O(mn^3)$ and $O(m^2n^2)$ respectively. We prove that their precision can be less than and equal to a constant on an instance.\\
\newline
\textbf{Overview. }Author will follow this format: defining objects, exploring the features of objects, proving algorithm, giving pseudocode
, computing runtime complexity of approach and finally present Exp. on instance. In these process, we give the discussion or summary to express author's viewpoint with problem. Then this paper is organized as follows. Firstly preliminary knowledge is in Section 2. Section 3 introduces basic definition, properties, method and proofs for graph, including pseudocode of BOTS and OBOTS, experiments. At the end give the solving problem of natural number BOCPS. The definition and method of graph partition are stated in Section 4. Similarly, there are pseudocode, algorithmic complexity and experiment. Section 5 proposes definition of edge and model of graph coloring. Finally we give the formula of graph coloring. Of cause, there show the algorithms BOGPC and BOERC with concerning complexity. We will evaluate those algorithms about probabilistic of exploring minimum chromatic value. The paper is concluded by a summary, a conjecture of \emph{Russell Paradox} and future work in Section 6.

\section{Preliminaries}
In this paper, we are interested in the connected graph. For each vertex $u$ on instance, there can be at least a path between $u$ and the others. We set each vertex can be labeled with number. Let $V$ be a collection of vertices having $V=\{v_1,v_2,\cdots ,v_n\}$.  We reserve the letter \emph{n} and the term $\vert V\vert $ for the number of vertices on an instance.\\
~\newline
\textbf{Partition of a set.}\cite{3} Given a no-empty universal set $A$,  there exists a family of sets $\bar{A}$, which is the partition of $A$, if and only if these following conditions hold:
\begin{enumerate} 
\item $\bigcup_{a_i\in \bar{A}} \bar{A}=A$
\item $a_i\neq\varnothing\quad \text{iff }a_i\in \bar{A}$
\item $\text{if }a_i, a_j\in \bar{A}~\text{and }a_i\neq a_j~\text{then }a_i\cap a_j=\varnothing $
\end{enumerate}

~\newline
\textbf{Equivalent Class Partition.}\cite{4} If there is a binary relation $\rho$ on set $A$, then there is a unique partition set $\bar{A}$ of the set $A$. For each component $a\in \bar{A}$, such that there are properties of reflexivity, symmetry and transitivity among all elements in set $a$ with respect to $\rho$.\\
~\newline
\textbf{Cartesian product.}\cite{5} Given $N$ sets $A_1,~A_2,~\cdots,~A_N$, there exists a multiplying sets $A_1\times A_2\times \cdots\times A_N$ and return \emph{N-ordered} vectors set $A^N$, in which for these members $v(i,j)$ such that $ v(i,j)\in A_i$.

\section{Graph Traversal}
\subsection{Definition and Property}
\begin{definition}
Given a no-empty and connected graph $G$ and vertices set $V$. Consider a pair $u,v\in V$. If there is a binary relation to characterize a behavior of traversal from vertex $u$ to vertex $v$. We define the binary relation as traversal relation, denote by $\tau$. We write the form $u\tau v$ and $(u,v)\in\tau$ to represent this relation on pair $u,v$. %

The ordered pair $(u,v)$ denotes a direction of left to right. We reserve the notation $\tau(1)$ equal to the first member in ordered pair, and then one equals to $\tau(2)$.

\end{definition}

\begin{theorem}\label{t1}
 If there is a no-empty traversal relation $\tau$ on an instance $G(V)$, then $\tau \subseteq V^2$.
\end{theorem}

\begin{proof} 
Given an instance $G(V)$. Let $\tau$ be a traversal relation on instance. As described in definition of Cartesian product, for each pair $u,~v\in V$ such that there is $(u,v)\in V^2$. Observe if pair $(u,v)\in \tau$ then $\tau\subseteq V^2$.%

Assume that there is a pair $(u,v)\in (\tau\setminus V^2)$. Then we have $u,v\notin V$, a contradiction to definition of traversal relation.\\
\end{proof}

It is obviously that there are some properties in traversal relation as follow:
\begin{enumerate} 
\item Reflexivity: if $u\in V$, then there may be $\exists (u,u)\in\tau$.
\item Anti-symmetry: if $(u,v)\in\tau$, then there may be $(v,u)\notin\tau$.
\item Anti-transitivity: if $(u,v),(v, t)\in\tau$, then there may be $(u,t)\notin\tau$.
\end{enumerate}

\begin{proof} 
Let $\tau$ be a traversal relation on instance $G(V)$. For each vertex $u\in V$ such that there may be $(u,u)\in V^2$ with definition of Cartesian product. Then observe that there may be pair $(u,u)\in \tau$ with $\tau\subseteq V^2$ as Theorem\ref{t1}. Hence, $\tau$ may have property of reflexivity.%

Consider a pair $(v,u)\notin\tau$ and there is $(u,v)\in\tau$. Observe that we can not say $(v,u)\in\tau$ hence, there is no property of symmetry in $\tau$.%

When there are three vertices $u,v,t\in V$ and pair $(u,t)\notin\tau$, similarly we can not say pair $(u,t)\in\tau$ with pairs $(u,v),(v,t)\in\tau$. We prove there is no transitive relation in $\tau$ and finish this proof.\\
\end{proof}

\begin{definition}
 Let $s$ be a subset of set $\tau$. If $\vert s\vert\geq 2$, for two arbitrary components $\tau_i,\tau_j\in s$ such that $\tau_i(1)=\tau_j(1)$. We call  set $s$ \emph{unit subgraph}. $S$ denotes  the collection of unit subgraphs.%

We reserve the subscript of set $s$ equal to the one of first member of each component in set $s$.\\
\end{definition}

\begin{theorem}\label{t2}
 Let $S$ be a collection of unit subgraphs on a traversal relation $\tau$. Then set $S$ is the partition of set $\tau$.
\end{theorem}

\begin{proof} 
Let $S$ be a collection of unit subgraphs on traversal relation $\tau$. We aim to prove three conditions hold for set $S$ on set $\tau$. Hence, first we can let set $S=\{s_1,s_2,\cdots,s_N\}$. Consider there is an isolated vertex $v_i\in V$. It is certainly that pairs $(\varnothing, v_i),(v_i,\varnothing)\notin \tau$ with the definition of traversal relation. Hence for each component $s_i\in S$, there is no such case $(\varnothing, u),(u,\varnothing)\in s_i$ with $s_i\subseteq \tau$ by the definition of unit subgraph. %

Consider a pair $(v_i, v_j)\in\tau$. As the definition of unit subgraph, there naturally may have a component $s_i\in S$ such that $v_i\tau v_j\in s_i$, which subscript is $i$. Hence, we have that $\tau \subseteq S$ and $s_i\neq\varnothing$.%

Assume $\tau\setminus S\neq\varnothing$ and having a pair $v_k\tau v_t\in (\tau\setminus S)$. There certainly may exist a component $s_k\in S$ and introduce pair $v_k\tau v_t$ to $ s_k$, thus observe there may be $v_k\tau v_t\in S$, a contradiction. Hence $\tau\setminus S=\varnothing$ and $\tau= S$.%

Consider two components $s_i,s_j\in S$ with $i\neq j$ such that $s_i\neq s_j$. Assume to $s_i\cap s_j\neq\varnothing$. Set a pair $v_k\tau v_t\in(s_i\cap s_j)$. As the definition of unit subgraph, observe there can be $k=i=j$, contradicts the given condition of $i\neq j$. Hence $s_i \cap s_j = \varnothing$.%

To sum up above, set $S$ is the partition of set $\tau$.\\
\end{proof}

\begin{lemma} \label{t3}
Unit subgraph is a Cartesian product set.
\end{lemma}

\begin{proof} 
Let $s_i$ be a unit subgraph. We have a term to characterize it as follow
\[s_i=\{v_i\tau v_1,~v_i\tau v_2,~\cdots ,~v_i\tau v_t\}:1\leq i,t\leq n.\]
The form can be written as follow
\begin{equation}
s_i=\{v_i\}\times \{v_1,v_2,\cdots,v_t\}:1\leq i,t\leq n. \label{subgraph}
\end{equation}
Observe that set $s_i$ is a Cartesian product set.\\
\end{proof}

\begin{definition}
As the form\eqref{subgraph}, we call term $\{v_i\} $ \emph{root set}, denote by $R(s_i)$. Call the right set \emph{leaf set}, denote by $L(s_i)$. Therefore, the unit subgraph $s_i$ can be abbr. by $s_i = R(s_i)\times L(s_i)$.
\end{definition}
~\newline
\textbf{Claim.} The cardinality of a multiple set is the number of difference members, not be the quantity of members. We reserve the notation $\{x\}^{\hash\omega}$ or $(x)^{\hash\omega}$ to represent a component  in a multi-set, the $\omega$ is the count of element $x$ and $\omega\geq 0$. We call \emph{group} for a component containing same and repeated elements. We define the group minus as that $x^{\hash s}\setminus x^{\hash t}=x^{\hash s-t}$, if and only if $s\geq t$. Then the difference value is 0 If $s\leq t$.%
\begin{definition}
If there is a multiple set $\tau_m$ and $\hash \tau_m=\hash\tau$, i.e. for each pairs $u\tau v\in\tau_m$ such that $u\tau v\in\tau$, then we call set $\tau_m$ \emph{multiple traversal relation}. For convenience, we use $(u,v)^{\hash\omega}$ denote each group in set $\tau_m$. The notation $\omega$ is the count of the pairs $(u,v)_{\in\tau_m}$.\\
\end{definition}

\begin{definition}
Let $g_i$ be a subset of multiple traversal relation $\tau_m$. If $\vert g_i\vert \geq 2$, then for two groups $\tau_i^{\hash \omega_i},\tau_j^{\hash \omega_j}\in g_i$ such that $\tau_i(1)=\tau_j(1)$. We call set $g_i$ \emph{weighted unit subgraph}, and reserve the subscript of set $g_i$ equals to the one of first element of each pair in set $g_i$. The collection of weighted unit subgraphs we denote by $g$.\\
\end{definition}

\begin{lemma}\label{t6}
Let $\tau_m$ be a multiple traversal relation and $g$ be a collection of weighted unit subgraphs on set $\tau_m$. Then set $g$ is the partition of set $\tau_m$.
\end{lemma}

\begin{proof} 
Let $g$ be a collection of weighted unit subgraphs on multiple traversal relation $\tau_m$. As the definition of weighted unit subgraph, for a group $(v_i,v_j)^{\hash\omega}\in \tau_m$ and $\omega\geq 1$, there may be $g_i\in g$ such that group $ (v_i,v_j)^{\hash\omega}$ may be introduced to set $g_i$. Hence observe $g_i\neq\varnothing$ if $\exists (v_i,v_j)^{\hash\omega}\in \tau_m$.%

If there is $\tau_m\setminus g\neq\varnothing$, then at least we can have a group $\tau_i^{\hash\omega}\in(\tau_m\setminus g)$. Summarizing above, It is certainly that there may be a component $g_i\in g$ such that group $\tau_i^{\hash\omega}$ can be introduced to $g_i$ hence, $\tau_m\setminus g=\varnothing$ and $\tau_m= g$.%

Consider two components $g_i,g_j\in g$ with $i\neq j$ such that $g_i\neq g_j$. Assume to $g_i\cap g_j\neq\varnothing$ and at least a pair $v_k\tau v_t\in (g_i\cap g_j)$. Then, we have $k=i=j$ a contradiction to given condition $i\neq j$ as described in definition of weighted unit subgraph. Hence $g_i\cap g_j=\varnothing$.%

For satisfying three conditions for set $g$ on set $\tau_m$, we understand set $g$ is the partition of set $\tau_m$.\\
\end{proof}

\begin{definition}\label{d1}
There is a no-empty traversal relation $\tau$ on graph $G$. Consider each pair $u,v\in V$ and a trail $T$ on instance. If the pair $u,v$ lies on trail $T$ with a constraint of direction respect to the traversal relation $\tau$, then we call this constraint \emph{traversal visiting}.\\
\end{definition}

\begin{theorem}\label{t4}
There is a no-empty multiple traversal relation $\tau_m$ on an instance G(V). Then set $\tau_m$ characterizes the traversal visiting among all vertices on instance.
\end{theorem}

\begin{proof} 
Given an instance $G(V)$. Let $\tau_m$ be a multiple traversal relation on it. Consider each pair $v_i,v_j\in V$. If there is a group $(v_i,v_j)^{\hash\omega_i} \in\tau_m$ and no group $(v_j,v_i)^{\hash\omega_j}\in\tau_m$, it is obviously that there is no traversal visiting on direction $(v_j,v_i)$, i.e. there is impossible for ordered pair $(v_j,v_i)$ to lie on each trail on instance. We call this pair \emph{directed graph}. %

Let $ (v_i,v_j)^{\hash\omega_i},(v_j,v_i)^{\hash\omega_j}\in\tau_m$. When $\omega_i=\omega_j=1$, we say there exists a bidirected traversal visiting between pair $v_i,v_j$; the case is a \emph{simple graph}. If $\omega_i=\omega_j$ and $\omega_i\cdot \omega_j>1$, then there are several bidirected and equal visiting to each other. Observe the case is a \emph{multi-graph}. For $\omega_i\cdot \omega_j> 1$ and $\omega_i\neq \omega_j$, then there may be unequal visiting opportunities between the pair. This instance is usually called \emph{mixed graph}.\\
\end{proof}

\begin{lemma}\label{t5}
There is $\tau\subseteq\tau_m$ and $S\subseteq g$ on an instance $G(V)$.
\end{lemma}

\begin{proof} 
Let $\tau_m$ be multiple traversal relation and $\tau$ be traversal relation on an instance. As the described in definition of multiple traversal relation, we have $\hash \tau_m=\hash\tau$. For each group  $(u,v)^{\hash\omega}\in\tau_m$ , there is $(u,v)=(u,v)^{\hash\omega}$, if and only if $\omega=1$. Then $\tau\subseteq\tau_m$, similarly prove $S\subseteq g$.\\
\end{proof}

With Theorem\ref{t4}, consider given a connected graph, there can be these data structures on it as follow:
\begin{align*}
 G& = (V,~\tau_m).\\
V&=\{v_1,~v_2,~\cdots,~v_n\}:1\leq n\leq n.\\
\tau_m&=\{(v_i,v_j)^{\hash\omega_{i,j}}\}:i,j\in\{\mathbb{N}\}\setminus\{0\};~ \omega_{i,j}\geq 1.
\end{align*}
~\newline
\textbf{Claim. }We reserve the abbr. $G=(V,~\tau)$ or $G=(V, ~\tau_m)$ to represent a connected instance with no-empty set $\tau$ or $\tau_m$ respectively.\\
~\newline
\textbf{Section Summary. }In this section, author constructs the basic logic for graph, that is $\tau\subseteq V^2$ with the property of reflexive. In following, author will gradually abstract the subset of traversal relation to construct new relation for problems. The new glossary \emph{unit subgraph} indeed is an equivalent class partition in set $\tau$, because there are three properties of reflexivity, symmetry and transitivity in relation of equal first element of each pair in unit subgraph. In set theory, the equivalent class also is the partition of set, but it is the unique partition of set respect to certain relation. Therefore it is the essential data structure for our research in this paper with the feature of uniqueness. %

Lemma\ref{t5} states a fact that all connected graphs can be viewed as an instance of mixed graph on traversal visiting, such that the exact graph traversal algorithms we will show has to cover all connected graph, which method is data-oriented only.\\

\subsection{Exact Traversal Algorithm}
There exist two demands for graph traversal, traversing vertices and traversing edges. In this paper, author only introduces the problem of  traversing vertices. Because the data of traversing edges are huge, and the method is similar to traversing vertices too.\\

\subsubsection{Definition}
As the definition\ref{d1} of traversal visiting, we define a characteristic function $\phi_{\tau}:V^2\rightarrow\{0,1\}$ as follow
\begin{equation*}
\phi_{\tau} (u,v)= \left\{
\begin{array}{ll}
1,\qquad&\text{iff }\exists (u, v)\in\tau .\\
&~\\
0,\qquad&\text{otherwise}.
\end{array}
\right.
\end{equation*}

Because of the case $\tau\subseteq\tau_m$ with Lemma\ref{t5}, therefore we must consider the group weight $\omega\leq 0$. Then the characteristic function $\phi$ will be converted to map the group with weight to the binary set as follow
\begin{equation*}
\phi_{\tau} (u,v)= \left\{
\begin{array}{ll}
1,\qquad&\text{iff } (u, v)^{\hash \omega}\in\tau_m \wedge \omega>0 .\\
&~\\
0,\qquad&\text{otherwise}. 
\end{array}
\right.
\end{equation*}

Now when the program enumerates the possible vertices, the program only needs to scan the leaf sets and checks the weights. The characteristic function $\phi_{\tau}$ provides a method of judgment to ensure enumerating valid vertices. Let $A$ be the subset of leaf set. The weighted unit subgraph is entry parameter. We define the enumerating operator $\Phi $ as follow
\begin{equation*}
\Phi (g_i)\coloneq \left\{
\begin{array}{ll}
v_j\in A,\quad&\text{if }\phi_{\tau}(v_i,v_j)_{v_j\in L(g_i)}=1.\\
&~\\
\text{undefine, }\quad&\text{otherwise. }
\end{array}
\right.
\end{equation*}

\begin{definition}
Let $\beta_i$ be a subset of set $\tau$. If $\vert \beta_i\vert \geq 2$. For two arbitrary pairs $\tau_i,\tau_j\in\beta_i$ such that $\tau_i(2)=\tau_j(2)$, then we call set $\beta_i$ \emph{visiting set}. The collection of visiting sets we denote  by $\beta$; and  reserve the subscript of component in set $\beta$ equals to the one of then element in each pair in native component.%

Use $T(\beta_i)$ to denote the set of all first elements in pairs and the then elements set is $H(\beta_i)$. It is easy to prove visiting set is a Cartesian product set like unit subgraph. \\
\end{definition}

\begin{theorem}\label{a1}
Let $\beta$ be a collection of visiting sets on set $\tau$. Then set $\beta$ is the partition of set $\tau$.
\end{theorem}

\begin{proof}
There is a collection $\beta$ of visiting sets on set $\tau$. As the described in definition of visiting set, for a pair $(v_i, v_j)\in\tau$, then there may exist a visiting set $\beta_j\subseteq\beta$ and $v_i\tau v_j\in\beta_j$. Hence, observe that for each component $\beta_i\in\beta$, there may be $\beta_i\neq\varnothing$ and $\tau\subseteq\beta$. %

Assume to $\tau\setminus\beta\neq\varnothing$ and at least a pair $\tau_k\in (\tau\setminus\beta)$. There may be a visiting set $\beta_k\in\beta$ such that pair $\tau_k$ can be introduced to set $\beta_k$. Hence there is a contradiction of $\tau_k\in\beta$ then, $\tau\setminus\beta=\varnothing$.%

Consider two components $\beta_i,\beta_j\in\beta$ with $i\neq j$ such that $\beta_i\neq \beta_j$. Assume to $\beta_i\cap \beta_j\neq\varnothing$ and a pair $v_k\tau v_t\in (\beta_i\cap \beta_j)$. With definition of visiting set, we can have $t=i=j$ a contradiction to given condition of $i\neq j$. Hence $\beta_i\cap \beta_j=\varnothing$%

Summarizing, we can understand set $\beta$ is the partition of set $\tau$. Indeed, the proof is in same fashion with unit subgraph, because they both are two equivalent classes on set $\tau$.\\
\end{proof}

\begin{definition}
Let $\nu_i$ be a subset of set $\tau_m$. If $\hash \nu_i\geq 2$. For two arbitrary groups $\tau_i^{\hash\omega_i},~\tau_j^{\hash\omega_j}\in\nu_i$, such that there is $\tau_i(2)=\tau_j(2)$. We call the set $\nu_i$ \emph{multiple visiting set}, the collection of multiple visiting sets denote by $\nu$. We reserve the subscripts of components in set $\nu$ equals to the ones of then elements in each group in native component. \\%
\end{definition}

\begin{theorem}\label{a2}
Let $\nu$ be a collection of multiple visiting sets on set $\tau_m$. Then set $\nu$ is the partition of set $\tau_m$.
\end{theorem}

\begin{proof}
Indeed we aim to prove three conditions hold for set $\nu$ on set $\tau_m$. With the same fashion of proof in Theorem\ref{t6}, it is clearly for us to prove the fact on set $\tau_m$. Here we need not do the repeated work again. \\
\end{proof}

\begin{definition}
Let $\hat{V}$ be a collection of Cartesian product of set $V$ on an instance $G=(V,~\tau_m)$, having $V^N\in\hat{V}$ with $1\leq N\leq \infty$. Let $P$ be a member in set $\hat{V}$. If for each ordered pair $x_i,x_{i+1}$ lies on sequence $P$, such that there is ordered pair $(x_i,x_{i+1})\in\tau_m$, then we call sequence $P$ \emph{connected path}.\\
\end{definition}

\begin{definition}\label{p}
Let $P$ be a connected path on graph $G=(V,~\tau_m)$ and $\nu$ be a collection of multiple traversal visiting sets. If there is such an approach of cutting graph, for each group $\{v_i\}^{\hash m}\in P$ such that   $\nu_i\coloneq\nu_i\setminus\{\beta_i\}^{\hash m}$, then we call this approach \emph{equivalent visiting}. Sequence $P$ is called \emph{equivalent visiting path}.\\
\end{definition}

With definition\ref{p}, we define the equivalent visiting operator $\psi$ as follow

\begin{equation*}
\psi (\nu_i)\coloneq \left\{
\begin{array}{ll}
\omega \coloneq \omega- 1,\quad&\text{if }\phi_{\tau}(v_j,v_i)_{v_j\in T(\nu_i)}=1.\\
&~\\
\text{undefine, }\quad&\text{otherwise. }
\end{array}
\right.
\end{equation*}

\begin{lemma}\label{a3}
Let $P$ be an equivalent visiting path on graph $G=(V,~\tau_m)$. There is an approach of equivalent visiting on path $P$. Consider a group $(v_x, v_i)^{\hash\omega}\in\tau_m$ and $v_i\in H(\nu_i)$.  Then the weight $\omega$ may be converged to 0 by invoking equivalent visiting operator $\psi$ if group $\{v_i\}^{\hash m}$ lies on $P$ and $m\geq \omega$.
\end{lemma}

\begin{proof}
Let $P$ be an equivalent visiting path on an instance $G=(V,~\tau_m)$, on which there is an approach of equivalent visiting. We set there is a component $\nu_i\in\nu$ with $v_i\in H(\nu_i)$. Consider a group $(v_j,v_i)^{\hash \omega}\in \tau_m$. If $\omega= 0$, we can understand that equivalent visiting operator $\psi$ do nothing inducing from the inner characteristic function $\phi_{v_j\tau v_i}= 0$, as definitions of these functions. \\%

If $\omega> 0$, then the $\omega$ may be self-subtract-one and returned by operator $\psi$ inducing from inner function $\phi_{v_j\tau v_i}=1$. Consider the vertex group $\{v_i\}^{\hash m}\in P$ with $m\geq 1$. As described in definition of equivalent visiting, this case can lead to $\nu_i\coloneq\nu_i\setminus(\beta_i)^{\hash m}$, such that there is $\omega(m)\coloneq\omega(m-1) -1$ by iteratively invoking operator $\psi$. When $\omega=0$, the function $\psi_{v_j\tau v_i}$ would return with nothing, thus we can understand that while for $\omega<m$, the function $\psi$ can not continue to compute $\omega$ as the inducing from function $\phi_{v_j\tau v_i}=0$. Hence, the $\omega$ can be converged to 0.\\

\end{proof}

\begin{theorem}\label{a4}
Let $\nu$ be a collection of multiple visiting sets and $P$ be an equivalent visiting path on instance. If for each component $\nu_i\in\nu$ and each group $\tau_k^{\hash\omega_k}\in\nu_i$ such that $\omega_k\geq 1$, then for each vertex group $\{u\}^{\hash m}\in P$ and $u\in H(\nu_i)$, we have $m\leq Max(\omega_k)$.
\end{theorem}

\begin{proof}
Let $\nu$ be a collection of multiple visiting sets and $P$ be an equivalent visiting path on instance $G=(V,~\tau_m)$. For each component $\nu_i\in\nu$, we let $D=\omega_1, \omega_2,\cdots, \omega_k$ represent all weights of groups in set $\nu_i$.%

Consider vertex $u\in H(\nu_i)$ and having group $\{u\}^{\hash m}\in P$ with $m\geq 1$. As the definition of equivalent visiting, there is $\nu_i \coloneq\nu_i \setminus (\beta_i)^{\hash m}$. Because the $\hash \nu_i=\hash \beta_i$, we have that $\omega_k(m)=\omega_k - m$ for $\omega_k\in D$. It implies the fore equation can be view as an iterative equation $\omega_k(s)=\omega_k(s-1) - 1$ with $1\leq s\leq m$. Let $M=Max(\omega_k)$. When $M\leq m$, for each $\omega_k(M)=\omega_k - M$ such that there is $\omega_k(M)\leq 0$. Then we have each $\omega_k(M)=0$ with Lemma\ref{a3}, and the number set $D$ converges at 0. For each weight of group $\tau_k$ is equal to 0, the enumerating function $\Phi$ can not introduce vertex $u$ to path $P$ as a valid vertex again. Then for each vertex $u$, such that at most there are $M$ possibilities on path $P$. Namely, $m\leq Max(\omega_k)$. If $m> Max(\omega_k)$, we can set $\beta_i$ is a constant, then $\nu_i=\varnothing$ with formula $\nu_i \coloneq\nu_i \setminus (\beta_i)^{\hash m}$. Hence, the case does not exist with respect to invoking enumerating functiuon $\Phi$.%

Because of this theorem shows the maximum possibility of visiting a vertex, author call it \emph{Equivalent Visiting Maximum Value Theorem}.\\
\end{proof}

\begin{lemma}\label{a5}
 Let $P$ be an equivalent visiting path on a finite and no-empty instance $G=(V,~\tau_m)$. Then path $P$ can be convergent.
\end{lemma}

\begin{proof}
Given a finite and no-empty graph $G=(V,~\tau_m)$. Let $P$ be an equivalent visiting path. Consider each group $\{u\}^{\hash m}\in P$ with $m\geq 1$, then there is $u\in H(\nu_i)$ such that $\nu_i\coloneq\nu_i\setminus (\beta_i)^{\hash m}$. Further we can understand the term equals to $\tau_m\coloneq\tau_m\setminus (\beta_i)^{\hash m}$ with set $\nu$ is partition of set $\tau_m$ as Theorem\ref{a2}. Let $F_{P}(\beta_i, m)=\tau_m\setminus (\beta_i)^{\hash m}$ and $\beta_i$ be a constant.

 We can see sequence $P$ as a discrete point-sequence for iterated function $F_{P}$, and have

\begin{equation*}
F_{P}(\beta_i, m)=\left\{
\begin{array}{ll}
\tau_m, & \text{iff } m=0.\\
~&~\\
F_{P}(\beta_i, (m-1)) - \beta_i, &\text{iff }m\geq 1 \text{ and } H(\beta_i)\in P.
\end{array}
\right.
\end{equation*}

Observe the function $F_{P}$ is iterative and monotone decreasing with input $\beta_i$. Because of instance being finite, therefore $F_{P}$ can be convergent at 0. Then path $P$ converges, to which function $\Phi$ can not introduce any vertex with all weights equal to 0.%

Similarly consider the end-node $v_i$ on current path $P$ with $v_i\in R(g_i)$. For each group $\tau_k^{\hash \omega}\in g_i$, if each $\omega$ equals to 0, even nor $F_{P}$ converges at 0, then path similarly converges.%

Assume the path is infinite. As the described in Theorem\ref{a4}, there are infinite weights, and then it implies that at least a number of a pair is infinity. Then, this assumption contradicts the condition of finite instance.\\
\end{proof}

\begin{lemma}\label{a6}
Let $P$ be an equivalent visiting path on graph $G=(V,~\tau_m)$ and $\nu$ be a collection of multiple visiting sets. Consider each maximum weight $M_i$ in each component $\nu_i\in\nu$. Then $\vert P\vert \leq \sum_{i=1}^{n}M_i$.
\end{lemma}

\begin{proof}
Given a finite and no-empty graph $G=(V,~\tau_m)$. Let $P$ be an equivalent visiting path on graph $G$. With Lemma\ref{a5}, the path $P$ is convergent. Consider each component $\nu_i\in\nu$, in which the maximum weight of group is $M_i$. With Theorem\ref{a4}, we understand that each vertex $v_i\in V$ can lie on path $P$ for at most $M_i$ possibilities. Hence there is $\vert P\vert = \sum_{i=1}^{n}M_i$, if $F_{P}=0$ as described in Lemma\ref{a5}.\\%

Consider there is such case, of which each weight of group in weighted unit subgraph equals to 0 as input for enumerating function $\Phi$. The function $\Phi$ can return with nothing. The path $P$ can be forced to converge even nor $F_{P}=0$. Its length can be shorter than $\sum_{i=1}^{n}M_i$. Hence, $\vert P\vert \leq \sum_{i=1}^{n}M_i$ it holds.\\
\end{proof}

\begin{lemma}\label{a7}
If there are self-cycles on an instance $G=(V, \tau)$, then they are invalid traversal visiting on equivalent visiting.
\end{lemma}

\begin{proof}
Given a no-empty graph $G=(V,\tau)$. Let $v_i\in V$ and pair $(v_i,v_i)\in\tau$. There is pair $v_i\tau v_i\in( s_i\cap \beta_i)$ and its weight equals 1 as the definitions of unit subgraph and visiting set. If there is an approach of equivalent visiting on instance. When the vertex $v_i$ lies on path $P$, as the definition of equivalent visiting, we have $\tau\coloneq \tau\setminus\beta_i$. Then the weight of pair $(v_i,v_i)$ will be forced to subtract 1, such that enumerating operator $\Phi$ can not introduce vertex $v_i$ again. Hence, the approach can not traverse self-cycle, i.e. the self-cycle equals to an empty traversal visiting on instance.\\
\end{proof}

\begin{descussion}
Here author briefly shows the viewpoint about the case of self-cycle on simple graph: the self-cycle at least is an instance of \emph{Russell Paradox}. Look at the term $v_i\tau v_i\in( s_i\cap \beta_i)$ above. When we partition the set $\tau$ with the relations of $\tau_i(1)=\tau_j(1)$ or $\tau_i(2)=\tau_j(2)$, we can not say the self-cycle is arriving or leaving on vertex. Thus the method of equivalent visiting on simple graph would filter out self-cycle as the symmetry relation of leaving and arriving.
\end{descussion}

\begin{definition}
On an instance, if an approach obtains those equivalent visiting paths depend on iteratively and alternately invoking enumerating function $\Phi$ and equivalent visiting function $\psi$ to enumerate vertices and modify the traversal relations, we call this approach \emph{Based On Table Search}, abbr. by \emph{BOTS}.
\end{definition}

\begin{lemma}\label{a8}
The approach BOTS can enumerate every equivalent visiting path on a simple graph $G=(V,\tau)$.
\end{lemma}

\begin{proof}
Given a no-empty graph $G=(V,\tau)$. Let $P$ be a connected path on graph. Consider a pair $(v_i,v_k)\in\tau$ with $v_i\neq v_k$. If vertex $v_i$ lies on path $P$, BOTS would invoke operator $\psi$ such that $\tau\coloneq\tau\setminus\beta_i$ with Lemma\ref{a5}. If having pair $(v_k,v_i)\in\tau$, then there can be $v_k\tau v_i\in\beta_i$, moreover its weight equals to 0. But for pair $v_i\tau v_k\notin\beta_i$ such that it weight is still 1. Hence, BOTS can introduce vertex $v_k$ to path $P$ by invoking the function $\Phi$ to scan the leaf set in subgraph $s_i$. Hence, function $\psi$ can not affect this work of function $\Phi$.%

With the weight of pair $v_j\tau v_i$($\forall v_j\in L(s_i)$) is equal to 0, while the BOTS scan the leaf set in subgraph $s_j$, it can not enumerate vertex $v_i$ as the next valid vertex. The method of equivalent visiting prevents repeated visiting, but not to block enumerating vertices, which weights are equal to 1.\\%

Consider the reasons of search terminating. Let vertex $v_i$ be current node on path $P$. If $L(s_i)\setminus P=\varnothing$, as the described in Theorem\ref{a4} the enumerating operator $\Phi$ can return with nothing, then the length of path $P$ may be less than or equal to $n$. If $\vert P\vert < n$, then the path is a dead-end path. Otherwise, the path is a Hamiltonian path. Namely, the BOTS may enumerate all connected paths on graph $G$. \\
\end{proof}

\begin{descussion}
Author gave two equivalent classes for traversal relation, so that we can define two operators on these classes. Their works can interactively constraint each other, according to the constraint of traversal visiting and method of equivalent visiting. Consequently, the approach of BOTS has no necessary to be a recursive method. Because there is not any demand of traversing a mixed graph or multi-graph, therefore author will not argument those problems in this paper. Hence all works stop at simple graph traversal. %
\end{descussion}

\subsubsection{BOTS Algorithm}
Here we will discuss the problem on the level of program. Summarize the augments above, we can obtain some conclusions as follow: The traversal relation can be organized as a table, in which unit subgraph can be a unit of data. We can evaluate the longest length of path. When the enumerating function returns an empty set, we know this search work on current path is over.%

Hence, we need define three set for approach as follow: First is the set \emph{Stack}, in which there are the path waiting for search. Set \emph{P} is second, which is a path containing a sequence of vertices in process of current exploring. And then set \emph{R} stores the final results and returns finally. The following pseudocode for the approach is given as Algorithm

\renewcommand\arraystretch{1.1}
\begin{longtable}{ p{120mm}}
\toprule  
\textbf{Algorithm 1: BOTS}\\
\toprule  
\textbf{input} graph $G=\{ g_1,g_2,\cdots,g_n\}$\\
\quad\quad \quad  set $\bar{G}=G$\\
\quad\quad \quad  set $P~\leftarrow v_i$\\
\quad\quad\quad set $ Stack$\\
\quad\quad\quad set $R=\varnothing$\\
\textbf{output}  $R$\\
\textbf{00}\quad                         \textbf{While}$(Stack\neq\varnothing$ or $P\neq \varnothing)$\\
\textbf{01}\qquad 			\textbf{If }$(P=\varnothing)$  \textbf{Than }$P=\emph{\textbf{pop}}( Stack);$\\
\textbf{02}\qquad                       \textbf{For} $1 \rightarrow \vert P \vert$ \textbf{do}\\
\textbf{03}\qquad   \quad          $\bar{G}\coloneq\psi(\nu_x) ~\Longleftarrow~\exists v_x\in P$;\\
\textbf{04}\qquad                      $A=\Phi(g_{s})~\Longleftarrow~ x_s=\emph{\textbf{end}}( P)$;\\        
\textbf{05}\qquad		        \textbf{If }$(A= \varnothing)$  \textbf{Than }$ R ~\leftarrow ~ P;~\bar{G}=G$;\\
\textbf{06}\qquad \quad           \textbf{Else }\textbf{For }$1\rightarrow \vert A \vert$ \textbf{ do}\\               
\textbf{07}\qquad \qquad		 $Stack~\leftarrow ~ P+v_i~\Longleftarrow~\exists v_i\in A$; \\ 
\textbf{08}\quad 	\textbf{output }$R=p_1,p_2,\cdots,p_N$;\\			  
\bottomrule
\end{longtable}
~\newline
\textbf{Algorithmic Complexity. }
The algorithmic complexity of BOTS focuses on the works of iterative reading and writing the table. In the process of search, program need record each info. of traversals, such that program update a table with size of $mn$(for $m=\vert T(\nu)\vert$). Consequently, in the worst-case  the loop times for a completed path is equal to the length of the final path. We have the gain of loop times for a path  is: $\lambda n$(for $\lambda\geq 1$) in worst-case. The runtime complexity of obtain a completed path is
\[T=\lambda n\cdot nm = \lambda m n^2.\]
When the instance is a \emph{r-regular} graph, for the $\lambda=1$ such that the runtime complexity is $O(mn^2)$. While the instance is a $K_n$ \emph{completed} graph, for the $m=n-1$, then the runtime complexity is $O(n^3)$. The probelm about how the algorithmic complexity for an instance will be discussed in following with optimized algorithm together.

\subsubsection{OBOTS Algorithm}
This subsection, author introduces the optimized algorithm without copying whole table, which author calls \emph{OBOTS}. The function $\Phi$ only need compare the weight with the count $k$ of a vertex appearing in current path, as Theorem\ref{a4}. The program needs not to copy, scan and modify the whole table. The pseudocode is in following
\renewcommand\arraystretch{1.1}
\begin{longtable}{ p{120mm}}
\toprule  
\textbf{Algorithm 2: OBOTS算法}\\
\toprule   
\textbf{input} graph $G=\{ g_1,g_2,\cdots,g_n\}$\\
\quad\quad \quad  set $P~\leftarrow v_i$\\
\quad\quad\quad set $ Stack~\leftarrow P$\\
\quad\quad\quad set $R=\varnothing$\\
\textbf{output}  $R$\\
\textbf{00}\quad                              \textbf{While}$(Stack\neq\varnothing)$\\
\textbf{01}\qquad                    $x_e= \emph{\textbf{end}}(P)~\Longleftarrow~ P=\emph{\textbf{pop}}( Stack);$ \\
\textbf{02}\qquad  		     $L(g_e)~\Longleftarrow~ g_e\in G ~\Longleftarrow~  x_e\in V$;\\        
\textbf{03}\qquad                    \textbf{For }$1~\rightarrow \vert L(g_e) \vert$ \textbf{ do}\\
\textbf{04}\qquad \quad                  $k= \emph{\textbf{ count}}(v_j\in P)~\Longleftarrow~\exists v_j\in  L(g_e) $;\\
\textbf{05}\qquad\quad                  $(\omega_j > k)$? $A~\leftarrow~v_j$ : \textbf{continue};\\
\textbf{06}\qquad 				\textbf{If }$(A=\varnothing)$  \textbf{Than }$ R ~\leftarrow ~ P$;\\                
\textbf{07}\qquad	\quad			\textbf{Else For }$1\rightarrow \vert A\vert$\textbf{ do}\\
\textbf{08}\qquad	\qquad			 $Stack~\leftarrow ~ P+v_i~\Longleftarrow~\exists v_i\in A$; \\ 
\textbf{09}\quad 		\textbf{output }$R=p_1,p_2,\cdots,p_N$;\\			
\bottomrule
\end{longtable}
~\newline
\textbf{Algorithmic Complexity. }
The loop times similarly is equal to the length of a completed path $\lambda n$. In per-loop, program need a comparing between two arrays with size of $m$ and $c\cdot n$(for $0<c\leq\lambda$). Hence, in worst-case the runtime complexity of a final path is
\[m\cdot\lambda n(\lambda n+1)/2 \leq m\lambda^2 n^2 .\]

The runtime complexity for a path is $O(mn^2)$, if instance is a \emph{r-regular} graph. And then $K_n$ \emph{completed} graph is $O(n^3)$. The advantage of OBOTS is the search job and table can be separated on different machines.\\
~\newline
\textbf{Problem Complexity. }
The enumerating function in BOTS or OBOTS, both returns a subset of leaf set. In the worst-case, there is $\vert A\vert = m$, which is the cardinality of leaf set. Consequently, the number of backtrace paths may increase quickly as exponential. For an instance, if we can use the quantity of paths to represent the complexity, then we reserve the word \emph{Breadth} to denote the search breadth of an instance. Hence, we have to write the complexity function is $O(mn^2B)$ for BOTS and OBOTS, which latter \emph{B} is the inherent breadth of an instance. Hence, because $mn$ equals to the quantity of pairs in set $\tau$, we can write the term as $O(n\vert\tau\vert B)$ for simple graph.

\subsubsection{Exp. of Graph Traversal}
In this subsection, there three experiments: First is to test whether the algorithm is valid or no for graph traversal. These objects for experiment are completed graph, which we have understood the $Breadth=(n-1)!$ for each instance. Second we test several familiar figures, and then we can learn how much the Hamiltonian path and Hamiltonian cycle. Third we explore the breadth of simple graph and how relevant is between breadth and shape of figure.%

The tests were executed on one core of an Intel Dual-Core CPU T4400 @ 2.20GHz 2.20GHz, running Windows 7 Home Premium. The machine is 64bit system type, clocked at 800MHz and has 4.00 GB of RAM memory. The executive program is on console, which compiled by C++11\footnote{Code::Blocks 12.11 IDE; http://www.codeblocks.org}. The implemented algorithm is OBOTS. %

The results are listed in following Table 3. Columns mean as follow: Columns $n$ indicate the number of vertices on an instance. Columns L.C. give the total of overall loops for searching. Columns B. show the \emph{Breadth} on an instance. Columns (1)/(2) mean the value that how much loops to obtain a path in average-case. The 5th columns list the actual runtime for testing on an object. We set $n=5,6,7,8,9,10,11,12$.\\

\begin{center}
\renewcommand\arraystretch{1.0}
\begin{longtable}{ p{10mm}|| p{24mm}| p{20mm}| p{24mm}| p{26mm}}
\caption{\textbf{$K_n$ Completed Graph Search}}\\
\toprule
n&L.C.(1)&B(2)&(1)/(2)&R.T.\\
\toprule  
n=5&65&24&2.708333333&4 ms\\ 
\midrule 
n=6&326&120&2.716666667&13 ms\\ 
\midrule 
n=7&1 957&720&2.718055556&24 ms\\ 
\midrule 
n=8&13 700&5 040&2.718253968&175 ms\\ 
\midrule 
n=9&109 601&40 320&2.718278769&1 s 440 ms\\ 
\midrule 
n=10&986 410&362 880&2.718281526&13 s 970 ms\\ 
\midrule 
n=11&9 864 101&3 628 800&2.718281801 &2 m 51 s 271 ms\\ 
\midrule 
n=12&108~505~112&39~916~800&2.718281826 &36 m 3 s 915 ms\\
\bottomrule
\end{longtable}
\end{center}

Each \emph{Breadth} on instance is $(n-1)!$ equals our expected number before. Because each objects is completed graph, therefore we know that the cardinality of set $A$ of subset of leaf set is $n-1$ after the first enumerating stage. and the search breadth increases to $n-1$. After the second stage, the search breadth increases to $(n-1)(n-2)$. Hence, we can learn that in the $k$th stage, the search breadth would be $P_{n-1}^k$. We can count the loops as follow:

\begin{align*}
C&=1+(n-1)+(n-1)(n-2)+\cdots+(n-1)\cdots 2 + (n-1)\cdots 1\\
~&~\\
&=P^0_{n-1}+P^1_{n-1}+P^2_{n-1}+\cdots+P^{n-2}_{n-1}+P^{n-1}_{n-1}
\end{align*}

While multiplying each term at the right side with $ (n-1)!/(n-1)!$, then

\[C = (1/P^{n-1}_{n-1}+1/P^{n-2}_{n-1}+\cdots+1/P^{1}_{n-1}+1+1)(n-1)!\]

Let

\[ E=1+1+1/P^{1}_{n-1}+1/P^{2}_{n-1}+\cdots+1/P^{n-1}_{n-1}\]

When in case $n\rightarrow \infty$, the expression $E$ is the approximative formulate of \emph{Euler constant}. Here are the \emph{Euler constant} $e = 2.718281828459$, which precision is 12 decimal after the dot. It is clear that the number is 8 decimal after the dot for $n=12$. The $Row_{(9,2)}, Row_{(9,3)}$ is the integers for ratio to \emph{Euler constant}. \\

Consequently, it proves our algorithm is valid for graph traversal. There are five figures for our testing in following. Each instance is simple graph, some is the grid figures and the $4th$ is the dodecahedron on 2D.\\

~\newline
\begin{center}
\includegraphics[width=120mm]{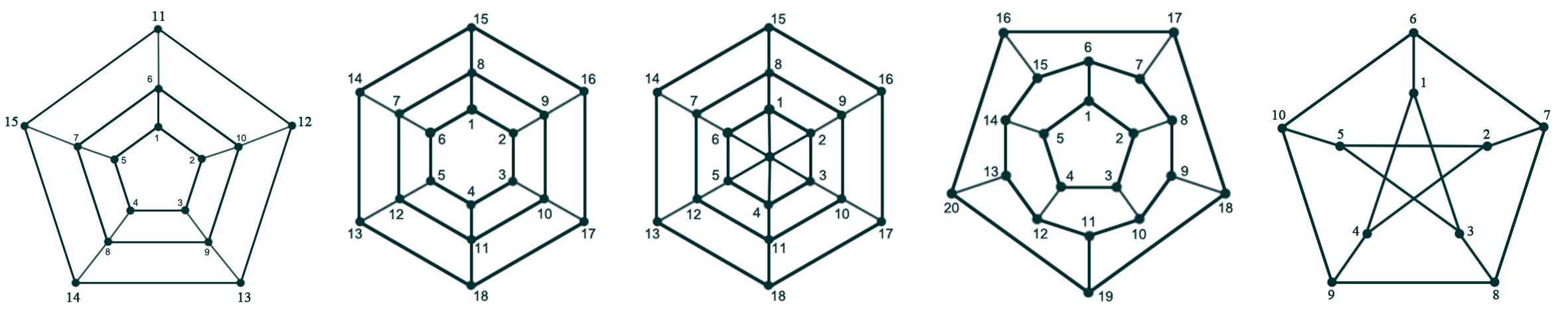}\\
~\\
\textbf{Figure 1}
\end{center}

In this experiment, we will record the loop times, breadth, actual runtimes in following Table 4: \\
\begin{center}
\renewcommand\arraystretch{1.0}
\begin{longtable}{ p{15mm}|| p{15mm}| p{18mm}| p{24mm}| p{26mm}}
\caption{\textbf{Simple Graph Exp. 1}}\\
\toprule
No.&L.C.(1)&B(2)&(1)/(2)&R.T.\\
\toprule   
1(n=15)&5 575&1 504&3.706781915&49 ms\\ 
\midrule 
2(n=18)&21 525&5 750&3.743478261&213 ms\\ 
\midrule 
3(n=19)&177 352&50 354&3.522123367&1 s 670 ms\\ 
\midrule 
4(n=20)&12 538&3 120&4.018589744&124 ms\\ 
\midrule 
5(n=10)&274&72&3.805555556&10 ms\\ 
\bottomrule
\end{longtable}
\end{center}

In following Table 5, Columns B shows the breadth of instance. Columns E.B gives the estimated value of breadth relevant with $n$. Columns H.P shows the quantity of Hamiltonian paths. Similar to Columns H.C there are the quantity of Hamiltonian cycles.\\
\renewcommand\arraystretch{1.0}
\begin{longtable}{ p{15mm}|| p{15mm}| p{22mm}| p{20mm}| p{26mm}}
\caption{\textbf{Simple Graph Exp. 2}}\\
\toprule
No.&B&E.B&H.P&H.C\\
\toprule  
1(n=15)& \quad 1 504&$\quad< 9\cdot P^2_{14}$&\quad 234&\quad 60\\ 
\midrule 
2(n=18)&\quad  5 750&$\quad<2\cdot P^3_{17}$&\quad 556&\quad172\\ 
\midrule 
3(n=19)&\quad  50 354&$\quad<11\cdot P^3_{18}$&\quad 4 390&\quad 1 008\\ 
\midrule 
4(n=20)&\quad  3 120&$\quad<10\cdot P^2_{19}$&\quad 162&\quad 60\\ 
\midrule 
5(n=10)&\quad  72&$\quad=P^2_9$&\quad 24&\quad 0\\ 
\bottomrule
\end{longtable}

Intuitively, we can get the conclusion of $Breadth=P^m_{n-1}$(for $m=L(g)$). Unlucky, it is wrong. The following experiment will give the new conclusion.\\

The object is a cycle-sequence, the quantity of cycles we denoted by $z$. For the first and the end cycle, on which there are $k$ vertices respectively. On each medium cycles, there are $2k$ vertices, such that on two adjacent cycles there are two $k$ vertices connect to each other. Hence, the cardinality of each leaf set is $m=3$, and the cardinality $\vert \tau\vert$ identically equals to $mn$. While change the number $z$, the number $m,n$ will not be changed. Indeed, for such type figure, if the $k=5\wedge z=3$, the figure would be the isomorphic to regular dodecahedron on 2D. In this experiment, we firstly let $n\in\{36,~44,~52,~60\}$ with $z=3$, then $k\in \{9,~11,~13,~15\}$ .

\renewcommand\arraystretch{1.0}
\begin{longtable}{ p{8mm}|| p{26mm}| p{15mm}| p{16mm}| p{36mm}}
\caption{\textbf{Simple Graph Exp. 3}}\\
\toprule
k=&B&H.P&H.C&R.T.\\
\toprule  
9& 804 226&1 412&300&1 m 2 s 986 ms\\ 
\midrule 
11&11 474 516&3 858&836&8 m 10 s 710 ms\\ 
\midrule 
13&158 293 248&10 258&2 080&1 h 57 m 46 s 50 ms\\ 
\midrule 
15&2 141 758 872 &27 044&1 130&28 h 33 m 59 s 672 ms\\ 
\bottomrule
\end{longtable}

Now we can evaluate $Breadth = CP_n^{\lceil k/z \rceil}$ from Table 6. Then we test the shape of figure to observe how the shape effects to the given instance for breadth and the numbers of Hamiltonian path and Hamiltonian cycle. The number $n$ would be $48$ not to change. Set the number $z\in\{3,4\}$. It leads to the number $k\in\{12,8\}$. The results is in Table 7 in following

\renewcommand\arraystretch{1.1}
\begin{longtable}{ p{8mm}|| p{26mm}| p{15mm}| p{16mm}| p{36mm}}
\caption{\textbf{Simple Graph Exp. 4}}\\
\toprule
z=&B&H.P&H.C&R.T.\\
\toprule  
3& 42 751 934&6 290&1 284&32 m 37 s 311 ms\\ 
\midrule 
4&56 999 358&7 272&656&41 m 36 s 134 ms\\ 
\bottomrule
\end{longtable}

It is obviously that the overall data structure has not been changed, but what it changed is the distribution of vertices on those cycles. Hence, the two actual runtimes is closed but the differences of other data are remarkable, so that we need more experiments to research. Here author have to stop going on this title.\\
~\newline
\textbf{Section Summary. }To sum up above, the advantage of BOTS algorithm is exact and general. You may unconcern what calss is given instance. On program level, this approach is simple, directed and easy to update. You can easily decompose it for different task of parallel or distributed computing. But the weakness is remarkable for that the approach lacks controlling for special works, so that it should make huge unnecessary data. Due to this reason, we need too research for the logic problems. The new problem arises that the approach randomly selects a vertex as start-node for traversal, and then what is the invariant in this problem.\\

\subsection{Trail, Path and Cycle}
\begin{definition}
Given a no-empty graph $G=(V,~\tau)$. Let $\hat{\tau}$ is a collection of $\tau^1, \tau^2,\cdots, \tau^N$ with $1\leq N\leq \infty$. Let $\pi$ be a component of set $\hat{\tau}$ and for each pair $\tau_i\in\pi$ such that $\tau_i(1)\neq\tau_i(2)$. If there exists $\vert \pi\vert>1$ and each ordered pair $\tau_i, \tau_{i+1}\in\pi$ such that $\tau_i(2)=\tau_{i+1}(1)$, then we call set $\pi$ \emph{trail}, the collection of trails is denoted by $\Pi$.\\
\end{definition}

\begin{definition}
Given a no-empty instance $G=(V,\tau)$. We may call $G$ \emph{connected}, If and only if for two arbitrary vertices $u, v\in V$ with $u\neq v$ such that there is at least a trail $\pi$ and $u,v\in\pi$. The number of ordered pairs on trail $\pi$ is called \emph{length}, denote by $\vert \pi\vert$. \\
\end{definition}

\begin{definition}
Consider a trail $\pi_k$. We call trail $\pi_k$ \emph{path}, if and only if for each vertex $u\in\pi_k$ such that
\[ u\in\bigcap_{1\leq N\leq 2}(\tau_{i})_{N}:i\in\{1,2,\cdots,\vert \pi_k\vert\}, \tau_i\in\pi_k.\]
The collection of paths  denote by $\vec{\pi}$.\\
\end{definition}

\begin{definition}
Consider a trail $\pi_k$. We call set $\pi_k$ \emph{cycle}, if and only if for each vertex $u\in\pi_k$ such that
\[ u\in\bigcap_{N=1}^2(\tau_{i})_{N}:i\in\{1,2,\cdots,\vert \pi_k\vert\}, \tau_i\in\pi_k.\]
The collection of cycles denote by $\mathring{\pi}$.\\
\end{definition}

\begin{definition}
Given a no-empty path $\vec{\pi}_k=(\tau_1,\tau_2\,\cdots , \tau_N)$. Let $A$ be a sequence set of vertices $x_1,x_2,\cdots,x_k$. If for each vertex $x_i\in A$ such that $x_i=\tau_i(2)$ with $1\leq i< N$ and $\tau_i\in\vec{\pi}_k$, then we call set $A$ \emph{medium vertices set} of path $\vec{\pi}_k$.\\
\end{definition}

\begin{theorem}\label{c1}
Let A be the medium vertices set of path $\vec{\pi}_k$. Consider pair $ x_i,~x_j\in A$ with $ i\neq j$, then $ x_i\neq x_j$.
\end{theorem}

\begin{proof}
Given a no-empty path $\vec{\pi}_k$. Let A be the medium vertices set of path $\vec{\pi}_k$ and pair $x_i,x_j\in A$ with $ i\neq j$. Observe there is $x_i=\tau_i(2)=\tau_{i+1}(1)$ and $x_j=\tau_j(2)=\tau_{j+1}(1)$ with the definition of medium vertices set and trail. We assume to $x_i=x_j$, thus we can have $x_i\in\bigcap\tau_k:k\in \alpha,~\alpha=\{i,~i+1,~j,~j+1\}$. If $j=i+1$, then $\vert \alpha \vert=3$. Otherwise $j\neq i+1$ having $ \vert \alpha \vert =4$, similar to number $j$. However all these cases contradict the definition of path. Hence, $x_i\neq x_j$ i.e. the set $A$ is not a multi-set.\\
\end{proof}

\begin{lemma}\label{c1_1}
Given a no-empty path $\vec{\pi}_j$ and $N=\vert \vec{\pi}_j\vert$. Let $A$ be the medium vertices set of path $\vec{\pi}_j$. If $N>2$, then $A\subseteq\vec{\pi}_j\setminus \{\tau_1,~\tau_N\}$.
\end{lemma}

\begin{proof}
Given a no-empty path $\vec{\pi}_j$ and $N=\vert \vec{\pi}_j\vert$. Let $A$ be the medium vertices set of path $\vec{\pi}_j$. Consider $N=2$. Let ${\pi}_j=(\tau_1,~\tau_2)$. Then $A\subseteq\tau_1\cap\tau_2$ and $\vec{\pi}_j\setminus \{\tau_1,~\tau_2\}=\varnothing$. Hence,  $A\subseteq\vec{\pi}_j\setminus \{\tau_1,~\tau_N\}$ it does not hold if $N=2$.%

Consider $N>2$. For each $\tau_i\in {\pi}_j$ and $1<i<N$, such that we have $\tau_{i-1}(2)=\tau_i(1)$ and $ \tau_i(2)=\tau_{i+1}(1)$ with the definition of trail. Observe there is pair $\tau_i\in (\tau_{i-1}\cup\tau_{i+1})$. Thus for $N=3$ such that there is pair $\tau_2\in (\tau_1\cup\tau_3)$. For $x_i,x_2\in A$, we have $x_1\in(\tau_1\cap \tau_2)$ and $ x_2\in(\tau_2\cap \tau_3)$. Hence, for pair $\tau_2=x_1\tau x_2$ having $A=\tau_2$ and $ A=\vec{\pi}_j\setminus\{\tau_1,\tau_3\}$, i.e. $A=\vec{\pi}_j\setminus\{\tau_1,\tau_N\}$ holds if $N=3$.%

We set a sequence $\vec{\pi}_k$ on path $\vec{\pi}_j$ having $\vec{\pi}_k\subseteq \vec{\pi}_j$, and $\vert \vec{\pi}_k\vert = k$ with $3<k<N$. Let $A_k$  be the medium vertices set of path $\vec{\pi}_k$. Moreover we can set $A_k\subseteq\vec{\pi}_k\setminus \{\tau_1,~\tau_k\}$ holds. Further, there may be a sub-sequence $\vec{\pi}_{k+1}$ on path $\vec{\pi}_j$ with $\vec{\pi}_{k+1}=(\vec{\pi}_k,~\tau_{k+1})$. With the definition of trail, there may be vertex $x_k\in (\tau_k\cap\tau_{k+1})$. Let $A_{k+1}$ be the medium vertices set of path $\vec{\pi}_{k+1}$, then there is $A_{k+1}=(A_k,~x_k)$. Observe that we have $x_k\in\vec{\pi}_{k+1}\setminus \tau_{k+1}$ such that $A_{k+1}\subseteq \vec{\pi}_{k+1}\setminus\{\tau_1,~ \tau_{k+1}\}$. Hence it $A\subseteq\vec{\pi}_j\setminus \{\tau_1,~\tau_N\}$ holds if $N\geq 3$. However $\tau_1(1)=\tau_N(2)$ or otherwise, we can have $\vert A\vert = N-1$.\\
\end{proof}

\begin{theorem}\label{c2}
Let $\vec{\pi}_j$ be a path. For $\vert\vec{\pi}_j\vert=N$ and $N\geq 3$, then there is  $\bigcap_{i=1}^N\tau_i=\varnothing$.
\end{theorem}

\begin{proof}
Given a path $\vec{\pi}_j$ with $\vert\vec{\pi}_j\vert = N$. If $N=2$, there is  $\tau_1\cap\tau_2\neq\varnothing$ with definition of trail,  $\bigcap_{i=1}^N\tau_i=\varnothing$ does not hold if $N=2$.%

Consider $N=3$ having $u\in(\tau_1\cap\tau_2)$ and $v\in(\tau_2\cap\tau_3)$. For $u,v\in V$, there is for $ u\neq v$, becauce otherwise there may be $u\in\bigcap_{N=1}^{3}\tau_N$, a contradition. Hence $\tau_1\cap\tau_2\neq\tau_2\cap\tau_3$. Assume to the vertex $u=\tau_1\cap\tau_3$, then $u\in\tau_3$ such that $u\in (\tau_1\cap\tau_2\cap\tau_3)$, a contradiction to the definition of path. Therefore we have the vertex $u\notin\tau_3$, similar to $v\notin \tau_1$. Hence $\tau_1\cap\tau_2\cap\tau_3=\varnothing$.%

We set there is a sub-path $\vec{\pi}_k\subseteq\vec{\pi}_j$ with $k>3$ and  $\bigcap_{i=1}^k\tau_i=\varnothing$ holds. Then there may be a sub-path $\vec{\pi}_{k+1}\subseteq\vec{\pi}_j$ and $\vec{\pi}_{k+1}=(\vec{\pi}_k,~\tau_{k+1})$. It is obviously to $\tau_{k+1}\cap\bigcap_{i=1}^k\tau_i=\varnothing$. Then term $\bigcap_{i=1}^{N}\tau_i=\varnothing$ holds if $N\geq 3$.%

If for a path $\vec{\pi}_j$ such that there may be $\bigcap_{i=1}^N\tau_i=\varnothing$(for $\tau_i\in\vec{\pi}_j$). We can write the form is $(\cdots((\tau_1\cap\tau_2\cap\tau_3)\cap\cdots)\cap\tau_N)$, so that we can understand $N\geq 3$.\\
\end{proof}

\begin{theorem}\label{c3}
Let $\vec{\pi}_k$ be a path on an instance $G=(V,\tau)$, then there is $1\leq\vert \vec{\pi}_k\vert \leq n$
\end{theorem}

\begin{proof}
Given a no-empty graph $G=(V,~\tau)$. Let $\vec{\pi}_k$ be a path on graph $G$. With the definition of path, there is $\vec{\pi}_k\in\Pi$. With the definition of trail, there is $\Pi\subseteq\hat{\tau}$ and $\tau^N\in\hat{\tau}$ with $1\leq N$. Then, there exists possibility of $\vert \vec{\pi}_k \vert = 1$ if $N=1$. Assume to $\vert \vec{\pi}_k \vert = 1$ and $\vec{\pi}_k = (\tau_1)$. For each vertex $u\in \tau_1$, there is $u\in\bigcap\tau_1$. This assumption satisfies the definition of path. Hence, $\vert \vec{\pi}_k\vert =1$ if $\vec{\pi}_k\in\hat{\tau}^N$ and $N=1$.%

Consider $N\geq 2$. Let $\vec{\pi}_k =(\tau_i)_{ i=1}^{N}$ and $A$ be the medium vertices set of path $\vec{\pi}_k$ having $A=(x_i)_{i=1}^k:k=N-1$, as the described in Lemma\ref{c1_1}. Then for $N>n+1$ such that there may be $\vert A\vert > n$. Because set $A$ is no a multi-set with the Theorem\ref{c1} therefore, there can be $A\setminus V\neq\varnothing$, i.e. having $\exists u\in (\vec{\pi}_k\setminus V)$, a contradiction. %

In case $N=n+1$, then $\vert A\vert = n$ such that have $A=V$ with Theorem\ref{c1}. Then for $A \subseteq \vec{\pi}_k\setminus \{\tau_1,~\tau_N\}$, such that at least there may be a vertex $u\in\tau_1$ and $u\notin (A\cup V)$, a contradiction.%

Hence, we only consider the case $N=n$ and $\vert A\vert = n-1$. There may be a pair $u,v\in (\tau_1\cup\tau_N)$ and $u,v\in (V\setminus A)$ as the described in the Lemma\ref{c1_1}. Set $u=\tau_1(1)$ and $v=\tau_N(2)$. If $u\neq v$, then $V\subset\vec{\pi}_k$. Otherwise for $u= v$ having $V\setminus\vec{\pi}_k=\varnothing$. Hence we may have $u=v$ and $u\in(\tau_i\cap\tau_N)$ satisfied the definition of path, $\vert\vec{\pi}_k\vert \leq n$ holds and finish the proof.\\
\end{proof}

\begin{theorem}\label{c5}
Let $\vec{\pi}$ be a collection of paths and $\mathring{\pi}$ be a collection of cycles on an instance. Then $\mathring{\pi}\subseteq \vec{\pi}$.
\end{theorem}

\begin{proof}
Let $\Pi$ be a collection of trails. We let paths set $\vec{\pi}$ and cycles set $\mathring{\pi}$ have $\vec{\pi},\mathring{\pi}\in\Pi$. Consider each component $\mathring{\pi}_k\in \mathring{\pi}$. For each vertex $u\in \mathring{\pi}_k$ having $u\in\bigcap_{N=1}^2(\tau_i)_N~(\forall\tau_i\in \mathring{\pi}_k)$, with the definition of cycle. We can transform this term to $u\in\bigcap_{1\leq N\leq 2}(\tau_i)_{N}$(for $N=2$), satisfies the definition of path, i.e. cycle is an instance of path, $\mathring{\pi}\subseteq \vec{\pi}$.\\
\end{proof}

\begin{theorem}\label{c6}
Let $\mathring{\pi}_k$ be a cycle and $\mathring{\pi}_k=(\tau_i)_{i=1}^ N$. Then $N\geq 2$ and $\tau_1(1)=\tau_N(2)$.
\end{theorem}

\begin{proof}
Given a cycle $\mathring{\pi}_k$ and $\mathring{\pi}_k=(\tau_i)_{i=1}^ N$. If $N=1$, then for each vertex $u\in \mathring{\pi}_k$, it is obviously for only a pair in cycle not to satisfy the definition of cycle for each vertex in an intersection of two different ordered pairs. Then there must be $N>1$ for any cycle.%

With Theorem\ref{c5}, let $A$ be the medium vertices set of cycle $\mathring{\pi}_k$. With defintion, we have $A\subseteq\mathring{\pi}_k$. If set $A$ is no empty, then $\vert A\vert\geq 1$ and the minimum potential of set $A$ is 1. We set  $\vert A\vert=1$ and  $u\in A$ such that $u=\tau_1(2)=\tau_2(1)$. It is abviously that if exist a vertex $v=\tau_1(1)=\tau_2(2)$ and $v\neq u$, then have $v\in(\tau_1\cap\tau_2)$, satisfies the definition of cycle. Hence we understand there may be $\vert \mathring{\pi}_k\vert\geq 2$.%

Consider $N\geq 2$. We assume to $\tau_1(1)\neq \tau_N(2)$ and set $u=\tau_1(1)$ and $v= \tau_N(2)$. With definition of trail, there may be $u\neq\tau_1(2)$ and $v\neq \tau_N(1)$ such that pair $u,v\notin A$. We have $u\in\bigcap \tau_1$, a contradiction to definition of cycle, similar to vertex $v$. Hence, for cycle $\mathring{\pi}_k$, there is $\tau_1(1)=\tau_N(2)$ and  the length of minimum cycle is 2.\\
\end{proof}

\begin{definition}
Let $\mathring{\pi}_k$ be a cycle on a no-empty instance $G=(V,~\tau)$. For each vertex $u\in V$ such that $u\in\mathring{\pi}_k$, we call cycle $\mathring{\pi}_k$ \emph{Hamiltonian cycle}. The collection of  Hamiltonian cycles denote by $\mathring{\pi}_h$.\\
\end{definition}

\begin{theorem}\label{c8}
Let $\vec{\pi}$ be a collection of paths and $\mathring{\pi}_h$ be a collection of Hamiltonian cycles on graph $G=(V,~\tau)$. Set $\vec{\pi}_{max}$ is the subset of $\vec{\pi}$, in which the length of each component is $n$, then $\vec{\pi}_{max}=\mathring{\pi}_h$.
\end{theorem}

\begin{proof}
Given a no-empty graph $G=(V,~\tau)$. Let $\vec{\pi}$ be a collection of paths and $\mathring{\pi}_h$ be a collection of Hamiltonian cycles on graph $G$. Set $\vec{\pi}_{max}$ is the subset of $\vec{\pi}$, in which the length  is $n$ for each component. 
If $\vert V\vert =2$, we set each component $\vec{\pi}_k\in\vec{\pi}_{max}$ to $\vec{\pi}_k=(\tau_1,\tau_2)$,  then it is certainly that there can be at most two vertices in $\vec{\pi}_k$ and $\vert \vec{\pi}_k(i)\vert\leq 2$. Hence $\vec{\pi}_{max}=\mathring{\pi}_h$ it holds if $n=2$.  %

Consider $n>2$ and a component $\vec{\pi}_k\in\vec{\pi}_{max}$ with a medium vertices set $A$. With Lemma\ref{c1_1}, there is pair $\tau_1(1),\tau_n(2)\in (\vec{\pi}_k \setminus A)$ and $\vert A\vert=n-1$. We set $u=\tau_1(1)$ and vertex $v=\tau_n(2)$ such that having $u,v\in (\vec{\pi}_k \setminus A)$. Let set $B=A\cup\{u,v\}$. If $u\neq v$ then $\hash B=n+1$, we have $B\setminus V\neq\varnothing$ such that $\vec{\pi}_k\setminus V\neq\varnothing$, a contradiction. Otherwise, for $u=v$, such that $\hash B=n$ and $B=V$ having $\vec{\pi}_k\in\mathring{\pi}$. Namely, there may be $\vec{\pi}_{max}\subseteq \mathring{\pi}_h$ with definition of Hamiltonian cycle.%

Assume having a component $\mathring{\pi}_h'\in \mathring{\pi}_h$ with $\vert \mathring{\pi}_h'\vert=m$ and $m\leq n-1$. Let $A'$  be the medium vertices set of cycle $\vec{\pi}_h'$, then $\vert A'\vert\leq n-2$. Further we can set $B=A'\cup\{\tau_1(1),~\tau_m(2)\}$. With Theorem\ref{c6} for $\tau_1(1)=\tau_m(2)$, there is $\hash B\leq n-1$. It is obviously that $\mathring{\pi}_h'\notin \mathring{\pi}_h$ as definition Hamiltonian cycle. Hence for each component in Hamiltonian cycles set $\mathring{\pi}_h$, its length is impossibly less than $n$. Consequently, $\vec{\pi}_{max}=\mathring{\pi}_h$ holds. \\
\end{proof}

\subsection{Invariant in Graph Traversal}
\begin{definition}
Given a no-empty cycle $\mathring{\pi}_k$. There may be two sub-sequences $\alpha_1,~\alpha_2\in\mathring{\pi}_k$ and $\mathring{\pi}_k=(\alpha_1,~\alpha_2)$. If there is a permutation $\rho$ on cycle $\mathring{\pi}_k$, for $\mathring{\pi}_k'=\rho(\mathring{\pi}_k)$ having $\mathring{\pi}_k'=(\alpha_2,~\alpha_1)$, then we call it \emph{cycle permutation}. We set $\vert\alpha_1\vert = m$, and $M$ is count of permutation, cycle permutation denote by $\rho_m^M(\mathring{\pi}_i)$.%

 We call $M$ \emph{power} of $\rho$ and $m$ is called \emph{index} of permutation.\\
\end{definition}

\begin{theorem}\label{c4}
Let $\rho$ be a cycle permutation on a no-empty cycle $\mathring{\pi}_k$ having $\mathring{\pi}_k'=\rho_m^M(\mathring{\pi}_k)$, then $\mathring{\pi}_k'\in \mathring{\pi}$.
\end{theorem}

\begin{proof}
Given a cycle $\mathring{\pi}_k$ and $\vert\mathring{\pi}_k\vert\geq 2$. Consider $\vert \mathring{\pi}_k\vert = 2$ and set $\mathring{\pi}_k =(\tau_1,\tau_2)$. Then there is $\tau_1(1)=\tau_2(2)$ and $\tau_1(2)=\tau_2(1)$. It is obvious that no matter how much $M$ is, there is $\rho^M_{1\leq m\leq 2}(\mathring{\pi}_k)\in\mathring{\pi}$. It $\mathring{\pi}_k'\in \mathring{\pi}$ holds if $\vert \mathring{\pi}_k\vert = 2$.\\%

Consider $\vert\mathring{\pi}_k\vert =N$ and $2< N\leq n$. Let $\mathring{\pi}_k$ contain two sub-sequences $\alpha_1,\alpha_2$ and $\mathring{\pi}_k =(\alpha_1,\alpha_2)$ with $\vert \alpha_1\vert =m$. Set $\vert\alpha_2\vert =p$ such that $p=N-m$. When $M=0$ or $m=0$, we understand no changing on cycle $\mathring{\pi}_k$.%

Then first consider $M=1$. We may have $\mathring{\pi}_k'=\rho_m^1(\mathring{\pi}_k)=(\alpha_2,\alpha_1)$. Given each pair $\tau_s\in\mathring{\pi}_k$. We let the subscript $s$ represent the $s$th of position on sequence $\mathring{\pi}_k$. Then, after a cycle permutation, the pair $\tau_s$ is on the $j$th position of sequence $\mathring{\pi}_k'$. For $\tau_s\in\alpha_1$ having $1\leq s\leq m$. We can have the equation $j=s+p$ to calculate the position changing for each pair $\tau_s\in\alpha_1$ between two sequences. Similarly, when $\tau_s\in\alpha_2$, equation is $j=s-m$.%

Consider pair $\tau_m$ on sub-sequence $\alpha_1$ as end-pair, pair $\tau_m$ may lie on the $N$th position of sequence $\mathring{\pi}_k'$ by $j=m+p$. Similarly, pair $\tau_{m+1}\in \alpha_2$ is on the 1th position of sequence $\mathring{\pi}_k'$ by $j= m+1-m$. It is easy to have derivation $\tau_m(2)=\tau_{m+1}(1)\Longrightarrow\tau_1'(1)=\tau_N'(2)$(for $\tau_1',\tau_N'\in\mathring{\pi}_k'$). In the similar fashion, we can have $\tau_{p}'(2)=\tau_{p+1}'(1)$. Hence $\mathring{\pi}_k'\in\mathring{\pi}$.\\%

Let $M=h$. We set $\mathring{\pi}_{k_{h}}=\rho_m^h(\mathring{\pi}_k)$ and $ \mathring{\pi}_{k_{h}}\in\mathring{\pi}$ holds. If $M=h+1$, we can have $\mathring{\pi}_{k_{(h+1)}}=\rho_m^{1}(\mathring{\pi}_{k_h})$. Summarizing above, we can similarly prove $\mathring{\pi}_{k_{(h+1)}}\in \mathring{\pi}$. Hence, we understand $\rho_m^M(\mathring{\pi}_i)\subseteq\mathring{\pi}$.\\
\end{proof}

\begin{theorem}\label{c9}
If $\mathring{\pi}_i=\rho_m^M(\mathring{\pi}_i)$, then $M\leq\vert\mathring{\pi}_i \vert$.
\end{theorem}

\begin{proof}
Given a cycle $\mathring{\pi}_i$ and $\vert \mathring{\pi}_i\vert = N$. Let $\rho$ be a cycle permutation on cycle $\mathring{\pi}_k$ with index $m$ and power $M$ on it. It is obvious that if $N=2$, however $m=1$ or $m=2$, two cases both make $M\equiv 2$.%

Consider $N>2$. There are two sub-sequences $\alpha_1,\alpha_2$ on cycle $\mathring{\pi}_i$ such that 
\[\mathring{\pi}_i=(\alpha_1,\alpha_2):\vert \alpha_1\vert = m_1, \vert \alpha_1\vert = m_2.\]
These numbers $m, m_1, m_2, N$ have
\[  m_1=m,~ N=m_1+m_2.\]

As the described in Theorem\ref{c4} proof, for the positions changing of each pair between sequence $\mathring{\pi}_i$ and $\mathring{\pi}_i'$, we set number $a$ represent each subscript of pair on sequence $\mathring{\pi}_i$, which is the $a$th position on sequences. If $1\leq a \leq m$, then $a\coloneq a+m_2$. Otherwise, for $m< a\leq N$, there may be $a\coloneq a-m_1$. After  cycle permutating for $M$ times, if $a(M)=a$, we say $\mathring{\pi}_i=\rho_m^M(\mathring{\pi}_i)$, which all pairs come back to the initial positions. We assume the number $a$ pluses $m_2$ for $k_1$ times and subtracts $m_1$ for $k_2$ times, such that $a(M)\coloneq a+k_1m_2-k_2m_1$. Then while there is $k_1m_2-k_2m_1=0$, each pair comes back to initial position. Hence, we have an equation groups as follow: 
\begin{align*}
m_1&=m\\
k_1+k_2&=M\\
m_1+m_2&=N\\
k_1m_2-k_2m_1&=0
\end{align*}

Solving the equation groups, we obtain equation $Mm=k_1N$. If $m=1$, then $m_2=N-1$ such that $k_1=1,~k_2=N-1,~M=N$. There need $N$ times of cycle permutation, each pair can come back to the initial position.%

If $m=N$, then $m_1=N$ and $m_2=0$ such that $k_2=0$ and $k_1=M$. The value of $M$ can be any natural number, cycle $ \mathring{\pi}_i$ has not any changing. Similarly for $m=0$.%

If $1<m<N$ and $N$, $m$ are prime number to each other, then for $ N/m=M/k_1$ such that $M=N$ and $m=k_1$. There at least need $N$ times of cycle permutation.%

Consider $N,m$ have the greatest common factor, let $m=s\mu;~N=t\mu$ and $s < t$ such that $N/m=t\mu/s\mu=M/k_1$. Then there is $M=t$ and $M<N$, hence There need less than $N$ times of cycle permutation. Namely,  $M\leq\vert\mathring{\pi}_i \vert$ holds. We finish this proof.\\
\end{proof}

\begin{definition}
 Let $\xi_i$ be a sub-sequence on cycle $\mathring{\pi}_i$. If $\vert\xi_i\vert=\vert\mathring{\pi}_i\vert - 1$, then we call $\xi_i$ \emph{chain}, the collection of chains denote by $\xi$. If $\mathring{\pi}_i\in \mathring{\pi}_h$, then call it \emph{Hamiltonian chain}. The collection of Hamiltonian chains is denoted by $\xi_{h}$.\\
\end{definition}

\begin{theorem}\label{c10}
Let $\xi_i$ is a chain on a cycle $\mathring{\pi}_i$. For each vertex $u\in\mathring{\pi}_i$,
 there is $u\in\xi_i$.
\end{theorem}

\begin{proof}
Let $\xi_i$ be a chain on a cycle $\mathring{\pi}_i$. With the definition of chain, there is $\vert\xi_i\vert=\vert \mathring{\pi}_i\vert -1$. We set a pair $\tau_j\in (\mathring{\pi}_i\setminus \xi_i)$ and $\tau_j=(v_j,v_k)$. Moreover, we can set vertex $v_j\in(\tau_j\cap\tau_s)$ as the described in definition of cycle. Then, for $\tau_s\in( \mathring{\pi}_i\cap\xi_i)$ such that having vertex $v_j\in\xi_i$. It is Similar to vertex $v_k$. Hence, for each $u\in \mathring{\pi}_i$ such that $u\in \xi_i$.\\ 
\end{proof}

\begin{lemma}\label{c11}
Let $\xi_{h_i}$ be a Hamiltonian chain on an instance $G=(V,~\tau)$. Then for each vertex $u\in V$, there is  $u\in \xi_{h_i}$.
\end{lemma}

\begin{proof}
Given a no-empty graph $G=(V,~\tau)$. Let $\xi_{h_i}$ be a Hamiltonian chain on instance. With the definition of chain, there exists a Hamiltonian cycle $\mathring{\pi}_{h_{i}}$ such that $\xi_{h_i}\subseteq \mathring{\pi}_{h_{i}}$. With the Theorem\ref{c10}, for each vertex $u\in \mathring{\pi}_{h_{i}}$ have $u\in \xi_{h_i}$. Then as the definition of Hamiltonian cycle, we have that for each vertex $u\in V$ such that $u\in \xi_{h_i}$.\\
\end{proof}

\begin{theorem}\label{c12}
Let $\xi$ be a collection of chains on cycle $\mathring{\pi}_i$. If there is  $\xi\subseteq \rho_m^M(\mathring{\pi}_i )$, then $\vert\xi\vert \leq \vert\mathring{\pi}_i\vert$.
\end{theorem}

\begin{proof}
Given a cycle $\mathring{\pi}_i$. Let $\xi$ be a collection of chains on cycle $\mathring{\pi}_i$. If $\vert\mathring{\pi}_i\vert=2$, it is obviously that there are two chains on cycle $\mathring{\pi}_i$. With the Theorem\ref{c3}, $\vert\xi\vert = \vert\mathring{\pi}_i\vert$ holds if $\vert\mathring{\pi}_i\vert=2$.\\%

Consider $\vert\mathring{\pi}_i\vert=N$ and $N>2$, there is a cycle permutation $\rho$ on cycle $\mathring{\pi}_i$. The cycle $\mathring{\pi}_i$ can be such sequence $\mathring{\pi}_i=(\alpha_1,\tau_m,\alpha_2)$. Set sequence $B=(\alpha_1,\tau_m)$. When the power $M$ of cycle permutation $\rho$ equals to 1. There may be a new cycle $\mathring{\pi}_i'=(\alpha_2,B)=(\alpha_2,\alpha_1,\tau_m)$ as Theorem\ref{c4}. If there is a sub-sequence $\alpha=\mathring{\pi}_i'\setminus \tau_m$, it is naturally that sequence $\alpha$ is a chain on cycle $\mathring{\pi}_i'$ with definition of chain. With Theorem\ref{c9}, if $\mathring{\pi}_i=\rho_m^{M}(\mathring{\pi}_i)$ then $M\leq N$, hence $\vert\xi\vert \leq \vert\mathring{\pi}_i\vert$ holds.\\
\end{proof}

\begin{theorem}\label{c13}
Let $\mathring{\pi}_h$ be a no-empty collection of  Hamiltonian cycles on instance $G=(V,~\tau)$. Given each vertex $v_i\in V$. There is at least a component $\mathring{\pi}_{h_i}\in\mathring{\pi}_h$, which start-node is $v_i$.
\end{theorem}

\begin{proof}
Given a graph $G=(V,~\tau)$. Let $\mathring{\pi}_h$ be a no-empty collection of  Hamiltonian cycles on instance. Consider a component $\mathring{\pi}_{h_{k}}\in\mathring{\pi}_h$ with start-node is vertex $v_k$. With Theorem\ref{c8}, the length of $\mathring{\pi}_{h_{k}}$ is $n$. Then as Theorem\ref{c9}, if there is a cycle permutation $\rho$ with the power $M\leq n$ on cycle $\mathring{\pi}_{h_{k}}$. Hence, there may exist at most $n$ difference cycles in $\rho_m^n(\mathring{\pi}_{h_{k}})$, i.e. there are $n$ different pairs on the first position of those sequences.%

As definition of cycle, for any vertex $u$ in cycle, there can be vertex $u=\tau_t(1)$ and $u=\tau_s(2)$. As the definition of trail and path, for sequence $\mathring{\pi}_{h_{k}}$, there only exists a pair $\tau_t$ in which $u=\tau_t(1)$. 

Consider a vertex $v_i\in V$ with $v_i\neq v_k$. We have there exists a pair $\tau_i\in \mathring{\pi}_{h_{k}}$ and $v_i=\tau_i(1)$. Summarizing above, there exists a Hamiltonian cycle $\mathring{\pi}_{h_{i}}\in \rho_m^n(\mathring{\pi}_{h_{k}})$, on which first pair is $\tau_i$. Namely, there may exist at least a Hamiltonian cycle with start-node $v_i$ on instance.\\
\end{proof}

\begin{theorem}\label{c14}
Let $\mathring{\pi}_h$ be a collection of  Hamiltonian cycles. For each component $\mathring{\pi}_{h_j}\in\mathring{\pi}_{h}$, then for each cycle $\mathring{\pi}_{h_j}(s)\in\mathring{\pi}_{h_j}$ such that their start-nodes are same. Consider two components $\mathring{\pi}_{h_i},~\mathring{\pi}_{h_k}\in \mathring{\pi}_h$ with $\mathring{\pi}_{h_i}\neq\mathring{\pi}_{h_k}$.We have $\hash\mathring{\pi}_{h_i}=\hash\mathring{\pi}_{h_k}$.
\end{theorem}

\begin{proof}
Let $\mathring{\pi}_h$ be a collection of Hamiltonian cycles on instance $G=(V,~\tau)$. For each component $\mathring{\pi}_{h_j}\in\mathring{\pi}_{h}$, such that each cycle in component $\mathring{\pi}_{h_j}$, the start-node is same. Then for two components $\mathring{\pi}_{h_i},\mathring{\pi}_{h_k}\in \mathring{\pi}_h$ and having $\mathring{\pi}_{h_i}\cap\mathring{\pi}_{h_k}=\varnothing$. Set vertex $v_i\in V$ is the start-node of cycle $\mathring{\pi}_{h_i}$, similarly we can set $v_k$ is the start-node of cycle $\mathring{\pi}_{h_k}$.\\%

Consider having a cycle permutation $\rho$ with the power $M\leq n$ on cycle $\mathring{\pi}_{h_{i}}$. As Theorem\ref{c13}, for a cycle $\mathring{\pi}_{h_i}(s)\in \mathring{\pi}_{h_i}$, there may be a Hamiltonian cycles $\mathring{\pi}_{h_i}(s)_k\in \rho_m^M(\mathring{\pi}_{h_i}(s))$, on which the start-node is $v_k$. Then we have cycle $\mathring{\pi}_{h_i}(s)_k\in\mathring{\pi}_{h_k}$. Consequently, if $\hash \mathring{\pi}_{h_i}= K$, then there at least are $K$ Hamiltonian cycles in set $\mathring{\pi}_{h_k}$, with $\rho_m^M(\mathring{\pi}_{h_i}(j))\cap\mathring{\pi}_{h_k}\neq\varnothing$(for $1\leq j\leq K$).\\

Now we need prove the intersection above is $\mathring{\pi}_{h_k}$ self. Then We can assume a cycle $H_k\in\mathring{\pi}_{h_k}\setminus \rho_m^M(\mathring{\pi}_{h_i}(j))$. Similarly, there may be a cycle permutation $\rho$ with the power $M$ on cycle $H_k$, and there can be a Hamiltonian cycles $H_i\in\rho_m^M(H_k)$ and $H_i\in \mathring{\pi}_{h_i}$. For each pair $\tau_j\in H_k$ such that having $\tau_j\in H_i$. Hence, there may be $H_k\in\rho_m^M(\mathring{\pi}_{h_i}(j))$, a contradiction to assumption. Hence $\mathring{\pi}_{h_k}\subseteq\rho_m^M(\mathring{\pi}_{h_i}(j))$. \\%

Consider there may be two cycle $H_i(s), H_i(t)\in\mathring{\pi}_{h_i}$ with $s\neq t$ such that $H_i(s)\neq H_i(t)$. We assume there exists a cycle $H_k\in\mathring{\pi}_{h_k}$, and there is $H_k=\rho_{m_s}^{M_s}(H_i(s))$ and $H_k=\rho_{m_t}^{M_t}(H_i(t))$. Then for each pair $\tau_o\in H_k$ such that $\tau_o\in (H_i(s)\cap H_i(t))$. As the constraint from definition of trail, we may have $H_i(s)=H_i(t)$, a contradiction to given condition $H_i(s)\neq H_i(t)$. Hence there is no such cycle in component $\mathring{\pi}_{h_k}$, $\hash\mathring{\pi}_{h_i}=\hash\mathring{\pi}_{h_k}$ holds. We call the $K$ \emph{Invariant of Graph Traversal}.\\
\end{proof}

\begin{descussion}
As Theorem\ref {c14}, we understand that however you randomly choose a vertex on an instance as the start-node for graph traversal: the number of Hamiltonian cycles in result is equivalent to a constant. Certainly, if we find out those Hamiltonian cycles with a start-node, then those other Hamiltonian cycles with other start-node can be produced by cycle permutation. Theorem\ref{c12} shows that although there justly are Hamiltonian chains in results, but you can check the results is correct or wrong by this constant. And we know the real quantity of Hamiltonian cycles in result of BOTS: it is a half of number of results. Actually, cycles both are the invariant of graph traversal, but justly the Hamiltonian cycle is importance at present.\\%

Moreover, you can find the fact that cycle actually is the structure of \emph{Circular Linked List} in memory of machine. The chain actually is a sequence of data extracting from cyclic list, which keeps the order on a cycle.\\
\end{descussion}

~\newline
\textbf{Section Summary. }In this subsection, we introduce how to abstract the simple pairs relation to define trail, path and cycle, in which these logical relation is $\mathring{\pi}\subseteq\vec{\pi}\subseteq\pi\subseteq \hat{\tau}$. It leads to cycle not be a vertices multi-set again. The keypoint is that these definitions have been quantified. Underlaying the model, we can look for these things what we want. Hence you can see that some axioms in current theory can be proved in the new system. For the self-cycle, author has refused it in definition of trail at beginning. Because it can destroy the whole quantified structure, and for other reason, it is an invalid traversal visiting for equivalent visiting. \\%

If you insist to against this new logic structure, we show an instance to support this system. Let a sequence $A$ be a medium vertices set of cycle $B$. With Theorem\ref{c1}, we can write the vertices sequence set $C$ of cycle $B$ as follow $C=(u,~A,~u)$. This form is always to represent a cycle in current graph theory.\\%

For Theorem\ref{c6}, which claims the minimum length of cycle is 2. That is a truth. After you read the section of graph coloring, author believes you must agree this viewpoint.\\ 

\subsection{BOCPS Algorithm}
In this sub-section, author introduces a speed-up linear algorithm, which is searching the greatest common factor and the least common multiple for given two natural numbers. Recall the description in the block of Theorem\ref{c9} proof. Recall the equations group as follow%

\begin{align*}
a\coloneq a+k_1m_2-&k_2m_1\\
k_1m_2-k_2m_1 &= 0\\
m_1+m_2&=N\\
k_1+k_2&=M
\end{align*}

Consider two natural numbers $m_1,~m_2$. Let a cycle permutation $\rho$ be on a cycle $\mathring{\pi}$ with $\vert\mathring{\pi}\vert = m_1+m_2$. It can always make the $a$th pair come back to the initial position for less than and equal to $N$ times. \\%

It is obviously that $k_1,~k_2$ both record the process and they are the coefficients of $m_2,~m_1$ respectively and $k_1\cdot k_2\neq 0$. When the number $a$ in the block $\alpha_1$, there is $a\coloneq a+m_2$. Otherwise, if $a>m_1$, then $a\coloneq a-m_1$. Approach only controls the value of $a$ with $a<=m_1 \vee a>m_1$. While the $a$ equals to the initial value, we understand $k_1m_2-k_2m_1 =0$. For existing the greatest common factor on $m_1,~m_2$, then we know the $k_2,~k_1$ is the least ratio of integers on $m_1,~m_2$ as the form $N/m=M/k_1=(t\mu/s\mu)$. Hence, author call this algorithm \emph{Based On Cycle Permutation Search}, abbr. by BOCPS. The pseudocode is in following:\\

\renewcommand\arraystretch{1.0}
\begin{longtable}{ p{120mm}}
\toprule  
\textbf{Algorithm 3: BOCPS}\\
\toprule  
\textbf{input} $m_1,~m_2;$\\
\quad\quad \quad   $s=1;~ k_1=k_2=0;$\\
\textbf{output}  $k_1,k_2$\\
\textbf{00}\quad                         \textbf{For }$1\rightarrow m_1+m_2$ \textbf{ do}\\
\textbf{01}\qquad 			\textbf{if} $(s>m_1)  $\textbf{Than }$ s-m_1;~ k_2++;$\\
\textbf{02}\qquad \quad              \textbf{Else} $s+ m_2;~ k_1++;$\\       
\textbf{03}\qquad                       \textbf{if} $(s=1)$ \textbf{break;}\\                  
\textbf{04}\quad       		\textbf{output }$k_1,k_2$;\\                           			  
\bottomrule
\end{longtable}
~\newline
\textbf{Algorithmic Complexity. }As Theorem\ref{c9}, we understand the greatest loop times are $M$ and $M\leq N$. In the worst-case, if there is no common factor on $m_1,~m_2$. The program need to calculate for $4N$ times for plus or subtraction and for $2N$ times for logical comparison. When having a greatest common factor on $m_1,~m_2$, the loop times are $k_1+k_2$. Further if given two number $s_1, s_2$ having $s_1/s_2=k(m_1/m_2)$, then the loop times on $s_1,~s_2$ will not change with the $k_1,~k_2$ as the coefficients of $s_1, s_2$. Hence the runtime complexity is $\Theta(c) \leq T \leq O(n)$. In fact, the $n$ can be optimized equal to $max(m_1, m_2)/2$ if $max(m_1, m_2)/2>min(m_1,m_2)$.%

Summarizing above, this algorithm can solve four problems:
\begin{flushleft}
\begin{enumerate} 
\item Find the greatest common factor on two nature numbers $m_1,~m_2$: $F=m_1/k_1$ or $F=m_2/k_2$.
\item Find the least common multiple on two natural numbers $m_1,~m_2$: $L=m_1\cdot k_2$ or $L=m_2\cdot k_1$.
\item Find the minimumal ratio of integers for two natural number pair: $m_1/m_2= k_2/ k_1$.
\item Solving a natural number equations group:
\begin{align*}
k_1m_2-k_2m_1=0\\
k_1+k_2\leq m_1+m_2\\
k_1\cdot k_2\neq 0
\end{align*}
\end{enumerate}
\end{flushleft}

\section{Graph Partition}
Author aims to find out an approach to obtain an ordered distribution of vertices as some laws. The graph partition is the approach that those vertices are partitioned into a sequence of regions separately, in which they possess some properties like on coordinate. Consequently we can solve some problems underlaying those properties, e.g. the shortest path.

\subsection{Definition and Features}
\begin{definition}
Given a no-empty instance $G=(V,\tau)$. There is a family of vertices $\Sigma$ on set $V$. Let set $\Sigma$ be a sequence of components, we abbr. it by $\Sigma=(\sigma_i)_{i=1}^ t$. For each vertex $v\in \sigma_{i+1}$ with $\sigma_i,\sigma_{i+1}\in\Sigma$ and $1\leq i < t $, there may be at least a vertex $u\in\sigma_i$ such that $(u,v)\in\tau$. We call $\Sigma$ \emph{partition of graph}. Each components in set $\Sigma$ is called \emph{region}, denote by $\sigma$.%

If there are two regions $\sigma_i,\sigma_j\in\Sigma$, for $i<j$, we call $\sigma_i$ \emph{low region}, to contrary $\sigma_j$ is \emph{high region}.
\end{definition}

\begin{lemma}\label{gp2}
If we define the distance of each pair $u\tau v$ to one unit, then for two adjacent components on sequence $\Sigma$, the distance between them is one unit.
\end{lemma}

\begin{proof}
Let $\Sigma$ be a partition of graph with $\vert \Sigma\vert =N$ and $ N\geq 2$. Consider two neighbors $\sigma_i, \sigma_{i+1}\in\Sigma$. For each vertex $v\in\sigma_{i+1}$, there is always a vertex $u\in \sigma_i$ such that pair $(u, v)\in\tau$ with the definition of graph partition. When a traversal visiting is from $\sigma_i$ to $\sigma_{i+1}$, then at least there may exist a sub-path $u\tau v$ on the path. Let the distance of each traversal relation is one unit, we have the distance between two neighbors $\sigma_i, \sigma_{i+1}$ is one unit.%

Assume that the distance between pair $\sigma_i, \sigma_{i+1}$ is greater than one unit. Then at least there may be a sub-path $p$ with $\vert p\vert>1$ between $\sigma_i$ and $\sigma_{i+1}$ and from former to latter. Observe there may be more than and equal to 1 vertices on middle segment of path $p$ then, there at least may be a region $\sigma_k$ contains those vertices between regions $\sigma_i$ and $\sigma_{i+1}$ with definition of graph partition. This is a contradiction to given condition of $\sigma_i$ and $\sigma_{i+1}$ are neighbors on sequence $\Sigma$. Namely, there is no redundant vertex between two adjacent components on sequence $\Sigma$.\\
\end{proof}

\begin{lemma}\label{gp3}
Consider a pair $u\tau v$ on an instance with $u\neq v$. Let $\Sigma$ be a partition of graph and $\vert \Sigma \vert >0$. Consider ordered pair $(u,v)\in\tau$. If there is $u\in \sigma_i$ and $v\in\sigma_j$ with $j\geq i$ and $\sigma_i,\sigma_j\in\Sigma$, then $0\leq j-i\leq 1$.
\end{lemma}

\begin{proof}
Given a partition of graph $\Sigma$ and $\vert \Sigma \vert >0$. We set for a pair $u\tau v$, there are two regions $\sigma_i,\sigma_j\in\Sigma$ such that $ u\in\sigma_i$ and $v\in\sigma_j$ with $j\geq i$. Assume to $ j-i>1$, then for vertex $v\in\sigma_j$ such that having $u\in\sigma_{j-1}$ with definition of graph partition, a contradiction to assumption $j-i>1$. Hence, we have $0\leq j-i\leq 1$. Namely, there is no such a shortcut on sequence $\Sigma$ from $i$th low region to $j$th high region with $j-i>1$.\\
\end{proof}

\begin{lemma}\label{gp4}
Let $\Sigma$ be a partition of graph on instance $G=(V,\tau)$, then $\vert \Sigma \vert \leq \vert V\vert$.
\end{lemma}

\begin{proof}
Given a partition of graph $\Sigma$ on instance $G=(V,\tau)$. As the definition of partition of graph, set $\Sigma$ is the partition of set $V$, then $\vert \Sigma \vert \leq \vert V\vert$. Assume to $\vert \Sigma \vert=N$ and $N>n$. Then those $n$ vertices in set $V$ are partitioned into $N$ regions in set $\Sigma$. We may have two cases that some regions are empty or at least there exists an intersection of two regions is no empty. It is obviously that set $\Sigma$ is not a family of vertices in set $V$, contradicts the definition of partition of graph. Hence $\vert \Sigma \vert \leq \vert V\vert$. Observe in the worst-case, the length of sequence $\Sigma$ is $n$.\\
\end{proof}

\begin{definition}
If there is a pair of vertices in different regions separately, then we call them on $H$ direction. Otherwise, they are on $V$ direction.
\end{definition}

\begin{theorem}\label{gp6}
 Let $\Sigma$ be a partition of graph on instance with $\vert \Sigma \vert >0$. For each pair $u,v\in V$ with $u\neq v$, consider there are $u\in\sigma_i$ and $v\in\sigma_j$ respectively with $j\geq i$ and $\sigma_i,\sigma_j\in \Sigma$. If there exists a connected path $P$ from vertex $u$ to vertex $v$, then $\vert P\vert \geq \vert j- i\vert$. 
\end{theorem}

\begin{proof}
Let $\Sigma$ be a partition of graph on instance and $\vert \Sigma \vert >0$. Consider a pair $u,v\in V$ and $u\neq v$, there is $u\in\sigma_i$ and $v\in\sigma_j$ with $\sigma_i,\sigma_j\in \Sigma$. We set $P$ be a connected path from vertex $u$ to vertex $v$ with $P=(\tau_s)_{s=1}^t$. If $i=j$, it is certainly that $\vert P\vert \geq 0$.\\%

Consider in case $j>i$, we let $A$ be the medium vertices set of $P$, then $\vert A\vert = t-1$ with Lemma\ref{c1_1}. We assume to $t<j-i$, then $\vert A\vert < j- i-1$. Because there exist $j-i-1$ regions on the $H$ direction between $i$th and $j$th region, therefore there may exist such $\sigma_k$ having $A\cap \sigma_k=\varnothing$ with $ i\leq k\leq j$. We understand there is no such shortcut on path $P$ with Lemma\ref{gp3}. Hence, there is $t\geq j-i$.%

Consider each pair $\tau_h$ on path $P$ again. As Lemma\ref{gp2}, there may be such case each pair $\tau_h$ between two adjacent components in set $\Sigma$, then $\vert P\vert = j- i$. Hence, $\vert P\vert \geq j- i$.\\%

Given an instance of a cycle $\mathring{\pi_i}$ and $\vert \mathring{\pi}_i\vert=N$. Let the vertex $\tau_1(1)\in\sigma_1$ and $\vert \sigma_1\vert =1$. As the definition of graph partition, the other vertices can be partitioned into the $i$th region. It is naturally that there exists a connected path from vertex $\tau_1(1)$ to vertex $\tau_N(1)$. As the definition of cycle, there is the ordered pair $\tau_N\in\tau$ having $\tau_N(2)=\tau_1(1)$. The distance from $\tau_{N}(1)$ to $\tau_1(1)$ is only a pair. If $N\geq 3$, then we can understand that $\vert P\vert \geq j- i$ holds on the direction of low region to high region only.\\%

Hence $\vert P\vert \geq j- i$ is not a necessary condition in this Theorem.\\
\end{proof}

\begin{lemma}\label{gp7}
Given a partition of graph $\Sigma$ and $\vert \Sigma \vert >0$. On an instance $G=(V,\tau)$, for each neighbors $u,v\in V$ with $u\neq v$, such that there is a bidirected ordered pairs $(u,v),(v,u)\in\tau$. Consider a pair $u,v\in V$ and $u\neq v$, there are $u\in\sigma_i$ and $v\in\sigma_j$ with $i\leq j$. 
If there exists a shortest connected path $P$ from vertex $u$ to vertex $v$, then the inverted sequence of path $P$ is the shortest connected path from vertex $v$ to vertex $u$.
\end{lemma}

\begin{proof}
Given a partition of graph $\Sigma$ and $\vert \Sigma \vert >0$ on an instance with such case for each neighbors $u,v\in V$ with $u\neq v$, such that there is a bidirected ordered pairs $(u,v),(v,u)\in\tau$. Consider each pair $u,v\in V$ and $u\neq v$, there are $u\in\sigma_i$ and $v\in\sigma_j$ with $i\leq j$ and $\sigma_i, \sigma_j\in\Sigma$. We can set $P$ is a shortest connected path from vertex $u$ to vertex $v$. As the given condition for each ordered pair and its inverse both in set $\tau$ on instance, there is naturally an inverted sequence of path $P$ on instance. Let $\bar{P}$ be this inverted path. %

As described in Lemma\ref{gp3}, there is no shortcut on path $P$ from $i$th low region to $j$th high region. Then there is no shortcut on path $\bar{P}$ from $j$th high region to $i$th low region. Assume to exist a path $\bar{P'}$ form vertex $v$ to vertex $u$ and $\vert \bar{P'}\vert < \vert\bar{P}\vert$. Observe there at least exists a shortcut on path $\bar{P'}$. Let ordered pair $(x,~y)$ be such shortcut, on which there is $t-s\geq 2$ with $i\leq s,t\leq j$ such that vertex $x\in\sigma_t$ and $y\in\sigma_s$. As the given condition, there may be an ordered pair $(y,x)\in\tau$ and from $\sigma_s$ to $\sigma_t$. As Lemma\ref{gp3}, such ordered pair $(y,x)$ is no existing in set $\tau$. Hence, there is no such shortcut $(x,y)$ on path $\bar{P'}$, then $\bar{P'}=\bar{P}$. The path $\bar{P}$ is the shortest path from vertex $v$ to vertex $u$.\\%
\end{proof}

\begin{theorem}\label{gp8}
Let $\Sigma$ be a partition of graph on instance with $\vert \sigma_1\vert=1$. Consider a pair $u,v\in V$ and $u\neq v$ with $u\in\sigma_1$ and $v\in\sigma_i$. If there is a shortest path $P$ from vertex $u$ to vertex $v$, then $\vert P\vert = i-1$.
\end{theorem}

\begin{proof}
Let $\Sigma$ be a partition of graph on an instance $G=(V,\tau)$ with $\vert \sigma_1\vert=1$. For a pair $u,v\in V$ with $u\neq v$, We set  $u\in\sigma_1$ and $v\in\sigma_i$, and there is a  shortest path $P$ from vertex $u$ to vertex $v$. As Theorem\ref{gp6}, we have $\vert P\vert \geq i-1$. If $i=2$, then for $v\in\sigma_2$ such that having pair $(u,v)\subseteq\tau$ with the definition of graph partition. Observe $\vert P\vert = 1$ if $i=2$.%

Consider the case $i=3$, we may let  $p$ be the shortest path from vertex $u$ to vertex $v$ and $A$ be the medium vertices set of path $p$. Assume to $\vert p\vert >2$, then set $A$ at least contains two vertices $x_1, x_2$, and on path $p$ there are 4 vertices. Hence, there surely is the case $x_1, x_2\in\sigma_2$ or $x_2, v\in \sigma_3$. As definition of graph partition, there certainly is a vertex $x\in\sigma_2$ such that having $(x,v)\in\tau$ and $(u,x)\in\tau$. Obviously, the path $p$ is longer than the path $P=(u\tau x,x\tau v)$. Hence, $\vert P\vert = i-1$ holds while $i=3$.%

If $i>3$, we can set $\vert P\vert = k-1$ holds with $3<k$ and $ k=i-1$ on instance. It is obviously for each vertex $x_j\in\sigma_k$ having a shortest path $P(u,x_j)$ between $u$ and $x_j$, which length is $k-1$. As definition of graph partition, there is at least a vertex $x_i\in\sigma_k$ and the pair $(x_i, ~v)\in\tau$. Hence, there is $\vert P(u,v)\vert=k$ i.e., $\vert P\vert = i-1$ holds.%

With Lemma\ref{gp7}, if the instance is a simple connected graph, then inverted path $\bar{P}$ is the shortest path of vertex $v$ to vertex $u$.\\
\end{proof}
~\newline\textbf{Section Summary. }Recall those proofs above in this subsection, graph partition shows an ordered distribution started from start-region. For the shape of graph partition, the start-node is on the center-block with others wrapping it on the contour lines likes the relation of sun and planets. Each path on $V$ direction is a segment of counter line, and paths on $H$ direction are the connected path among those vertices on different counter lines. You can image the isomorphism of instance is a spider-net. 

Graph partition constructs a relation among these vertices likes Physical fields such that we can setup a model for analyzing the physical flow. But need to be careful of the directed graph. As the properties of traversal relation, there may be no reversing partition on it, such that the case cutting graph would happen on directed bridges.

\subsection{Algorithm of Partition}
The core idea of algorithm is simple. Consider a partitioned root $v_i$. Its leaves would be partitioned to the next region besides partitioned ones. The approach iteratively root-leaves-root partitions those vertices until no vertex to be partitioned. Then approach uses three set. First is set $R$ for storing the regions. $W$ is second set in which there are roots and leaves waiting for partition, and third is set $C$ including all candidates. Then these relations among them are $C=V\setminus R\wedge W \subseteq R$. Hence we define the partition function as follow:
$ \phi_{\tau}:L(s_i)\rightarrow P,~~\text{iff } u\in( L(s_i)\cap C) \text{;~otherwise undefine}$.

\renewcommand\arraystretch{1.0}
\begin{longtable}{ p{120mm}}
\toprule  
\textbf{Algorithm 3: Graph Partition}\\
\toprule  
\textbf{input} graph $G=s_1,s_2\cdots,s_n$\\
\quad\quad \quad  set $R=W~(\text{for }1\leq\vert W\vert<n)$\\
\quad\quad\quad set $C=V\setminus R$\\
\textbf{output}  $R$\\
\textbf{00}\quad                              \textbf{While}$(C\neq\varnothing)$\\
\textbf{01}\qquad                    \textbf{For }$1\rightarrow \vert W\vert$ \textbf{ do}\quad \\
\textbf{02}\qquad\quad  		     $s_i~\Longleftarrow ~v_i \in W $; \\        
\textbf{03}\qquad \quad                  $P\leftarrow\phi_{\tau}(L(s_i))$;\\
\textbf{04}\qquad\quad                   \textbf{If }$P=\varnothing$ \textbf{ Than  break};\\
\textbf{05}\qquad\qquad  		 \textbf{Else } $R(k)\leftarrow W;~ W\leftarrow P; ~C\setminus P;$ \\
\textbf{06}\quad 		\textbf{output }$R=\sigma_1, ~\sigma_2, ~\cdots,~ \sigma_k$;\\              
\bottomrule
\end{longtable} 
~\newline
\textbf{Algorithmic Complexity.} The program need scan whole table for $n$ times and read $m$ nodes per-time in worst-case. It can introduces $m$ vertices to next region such that set $C$ may be removed at most $m$ vertices from it per-time in nice-case. Hence, the set $C$ would be check $c$ times as follow:
\[ c=(n-1)+(n-m-1)+(n-m^2-1)\cdots+1\]
Consider choosing a vertex as start-node and $\vert C\vert = n-1$. If the instance is completed graph, then $m=n-1$. The equation can be written such $c=n-1$. There loop would run for only one time. Hence, the runtime complexity of this graph class is $(n-1)^2=O(n^2)$. %

If the instance is a \emph{r-regular} graph or tree, we understand the loop times can be less than $2\log_{m}n$. Such that the equation above can be transformed to:

\begin{align*}
 c=& 2n\log_{m}n -\sum_{i=1}^N m^i\quad N=\log_{m}n\\
=& 2n\log_{m}n -O(m^N)\\
\leq& 2n\log_{m}n \\
\end{align*}

Then the runtime complexity is $O(Cn^2)$(for $C=2m\log_{m}n$). If the instance is a directed path with $m=1$, then the runtime complexity is $O(n^3)$. We have the runtime complexity of graph partition is $O(n^3)$ in worst-case.

~\newline\textbf{Remark.} The sentence  $R=W~(\text{for }1\leq\vert W\vert<n)$ in pseudocode, it said that there may be $k$ vertices as a start-nodes group. Then it implies there are several sources in a complex physics field. And this graph partition may not associate with weights of edges. 

~\newline\textbf{Experiment of Graph Partition.} Those objects for experiment are on Figure 2 in following.
\begin{center}
\includegraphics[height=40mm]{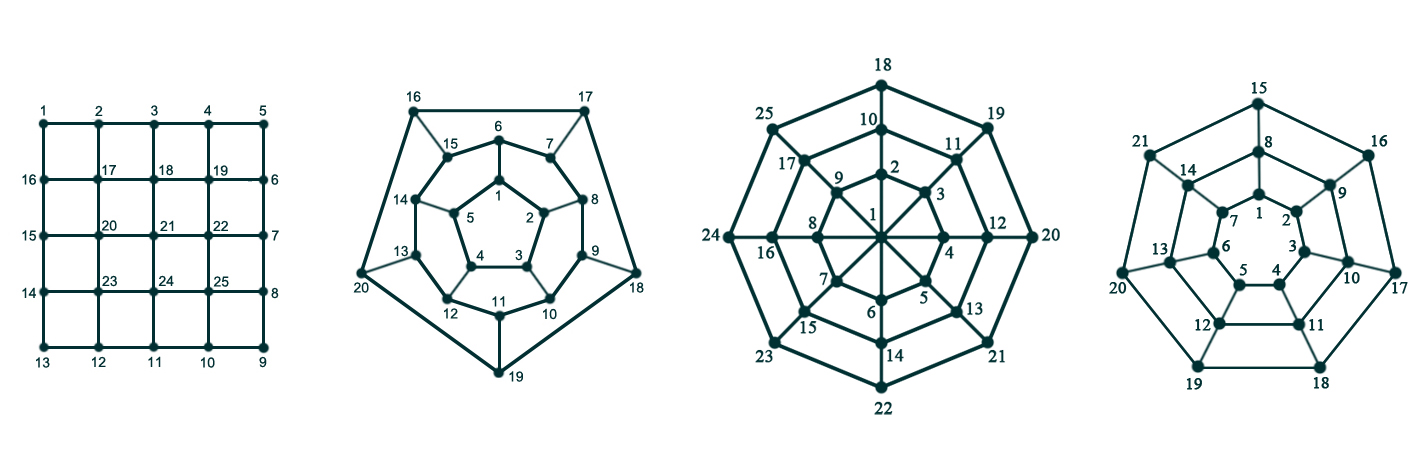}\\
\textbf{Figure. 2}
\end{center}

Author implemented this approach with PHP5.0+Apach2.0, for easy to compare the runtimes among them. We choose the $v_1$ on these figures as the start-node. The results are listed in following Table 10. Columns $n=$ show the number of vertices on each instance. Columns 2th contain two contents: first is the number of the components in set $\Sigma$, second is quantity of vertices in a relevant region respectively. Columns Loops record the times of loops. Columns R.T. are the runtime of program in practices. It is remarkable that all figures are simple graph.

\renewcommand\arraystretch{1.0}
\begin{longtable}{ p{18mm}|| p{53mm}| p{12mm} | p{20mm}}
\caption{\textbf{Exp. of Graph Partition}}\\
\toprule
\textbf{No.(n=)}&$\vert \Sigma\vert;\quad(\vert\sigma_i\vert)_{i=1}^N$&\textbf{Loops}&\textbf{R.T.}\\
\toprule  
1(n=25)& $\vert \Sigma\vert =9;\quad (1,2,3,4,5,4,3,2,1)$&\quad 9&1.3 ms\\ 
\midrule 
2(n=20)& $\vert \Sigma\vert =6;\quad (1,3,6,6,3,1)$&\quad 6&0.932 ms\\ 
\midrule 
3(n=25)&$\vert \Sigma\vert =4;\quad(1,8,8,8)$&\quad 4&2.821 ms\\ 
\midrule 
4(n=21)&$\vert \Sigma\vert =6;\quad (1,3,5,6,4,2)$&\quad 6&1.910 ms\\ 
\bottomrule
\end{longtable}

\section{Graph Coloring}
Graph coloring is that there is an approach to coloring each vertice on a given graph. The chromatic demand is each neighbors can not be labeled with the same colour. Then how much the least colors can label a given instance. Firstly, we give the description of problem with forms as follow 
\begin{flushleft}
\begin{enumerate} 
\item Consider a pair $u,v$ and $u\neq v$. There is a characteristic funtion $f$ to characterize the chromatic relation as follow
\begin{equation}\label{eqc1}
f(u,v)=f(v,u)=\left\{
\begin{array}{ll}
1,\quad&\text{ if  }\{ (u,v), (v,u)\}\cap\tau\neq\varnothing. \\
~&~\\
0,\quad&\text{ otherwise.}\\
\end{array}\right.
\end{equation}

\item  Given a pair $u,v$ and $u\neq v$. There is the relation of their chromatic values as follow
\begin{equation}\label{color1}
g(u,v)\coloneq\left\{
\begin{array}{lll}
C_u&\neq C_v,\quad& \text{if }f(u,v)=1.\\
~&~&~\\
C_u&=C_v \vee C_u,\neq C_v\quad &\text{if }f(u,v)=0.\\
\end{array}\right.
\end{equation}

\end{enumerate} 
\end{flushleft}

We call form\eqref{color1} \emph{coloring relation}. Obviously, there is a symmetrical relation between the pair values of two vertices $u,v$. If we solve this problem with the anti-symmetrical relation of graph traversal, then we can not solve the logical problem on graph coloring. Hence, we need renew to construct a new logic model for this problem.%

\subsection{Edge Relation}
\begin{definition}
Given a traversal relation $\tau$. Let $\epsilon$ be a subset of set $\tau$. For an ordered pair $(u,v)\in\epsilon$, such that there must be its inverse ordered pair $(v,u)\in\epsilon$. Then we call set $\epsilon$ \emph{edge relation}. The relation we  abbr. by $\{u,v\}\in\epsilon$ or $u\epsilon v$.
\end{definition}

Then we can obtain those properties of edge relation. 
\begin{flushleft}
\begin{enumerate} 
\item Anti-reflexivity: if $u\in V$, then there is $(u,u)\notin\epsilon$.
\item Symmetry: if $(u,v)\in\epsilon$, then there is $(v,u)\in\epsilon$.
\item Anti-transitivity: if $(u,v),(v,t)\in\epsilon$, then there may be $(u,t)\notin\epsilon$.
\end{enumerate} 
\end{flushleft}

\begin{proof}
Consider a pair $u,v\in V$. For $\tau\subseteq V^2$, then have $(u,v),(v,u)\in V^2$. If $u=v$, then that is a contradiction to set $V^2$, because otherwise set $V^2$ is a multi-set or exist two pairs $(u,u)$ in it with $(u,u)\neq(u,u)$. Hence, there is not any self-cycle in edge relation, i.e. edge relation possesses property of anti-reflexivity.%

The proof for other properties is similar to traversal relation. Here we have not to prove them.\\ 
\end{proof}

\begin{definition}
If there is $\tau\setminus \epsilon=\varnothing$ on an instance $G$, then we call this instance \emph{simple graph}, denote by $G_{\theta}$. We reserve the abbr. $G_{\theta}=(V,\epsilon)$ to represent a simple graph with no-empty edge relation.
\end{definition}

\begin{definition}
For each pair $u,v\in V$ on an instance $G_{\theta}=(V,\epsilon)$ with $u\neq v$, then there is pair $\{u,v\}\in\epsilon$. We call the instance \emph{completed regular graph},  denote by $G_{\eta}$. We specially define the isolated vertex is a completed regular graph.
\end{definition}

\begin{theorem}\label{cg1}
Given a no-empty graph $G=(V,\tau)$ and $G\notin G_{\theta}$, on which there is a characteristic function $f_{\tau}: V^2\rightarrow \{1,0\}$. If there is a function $g$ with a set $\overline{G}=(\overline{V},\bar{\epsilon})$ as codomain, such that $g\cdot f_{\tau}:V^2\rightarrow \bar{\epsilon}$, then $\overline{G}\subseteq G_{\theta}$. 
\end{theorem}

\begin{proof}
Given a no-empty graph $G=(V,\tau)$, set $G\notin G_{\theta}$. on which there is a characteristic function $f_{\tau}: V^2\rightarrow \{1,0\}$. Consider a pair $u,v\in V$. There is an ordered pair $(\{u,v\},1)_{\in f}$ as the value of function $f$, if and only if $\{u,v\}\cap\tau\neq\varnothing$. Otherwise, there is $ (\{u,v\},0)_{\in f}$.

We set $g$ is a function with a codomain $\overline{G}=(\overline{V},\bar{\epsilon})$, such that there is $g\cdot f_{\tau}:V^2\rightarrow \bar{\epsilon}$ defined in following

\begin{equation*}
g(f(u,v))\coloneq\left\{
\begin{array}{ll}
\{s,t\}\in \overline{\epsilon},\quad&\text{ if  }f(u, v)=1 \wedge\exists s,t\in\overline{V}.\\
~&~\\
\{s,t\}\notin \overline{\epsilon}, \quad&\text{ if  }f(u, v)=0\wedge\exists s,t\in\overline{V}.\\
\end{array}\right.
\end{equation*}

Then given a pair $u\tau v$ on graph $G$, that is mapped to an image in set $\bar{\epsilon}$ by the the composition function $g\cdot f$. If the ordered pairs $(u,v),(v,u)\notin\tau$, then no image in set $\bar{\epsilon}$. We can observe for each pair $\{s,t\}\in \epsilon$ such that having  $\bar{\epsilon}=\bar{\tau}$ on graph $\overline{G}$. As the definition of simple graph, there is  graph $\overline{G}\in G_{\theta}$. With the coloring relation\eqref{color1}, if exist a function $\rho$ and $\rho:V\rightarrow \overline{V}$, then there is $u\mapsto s\wedge v\mapsto t$, we can say the $\overline{G}$ is a relation model of chromatic values on graph $G$.%

Consequently, we understand that the simple graph is the sufficient condition of graph coloring problem for a given instance.\\
\end{proof}

\begin{definition}
Let $\lambda$ be a subset of edge relation $\epsilon$ with $\vert \lambda\vert = N$. Consider a vertex $u\in\lambda$. For each component $\epsilon_i$ in set $\lambda$, there is $u\in\epsilon_i$. We call set $\lambda$ \emph{edge relation subgraph},  the collection of subgraphs denote by $\Lambda$.
\end{definition}

\begin{lemma}\label{cg2}
The edge relation subgraph is a Cartesian product set.
\end{lemma}

\begin{proof}
With the definition of edge relation subgraph, we set edge relation subgraph $\lambda$ having $\vert \lambda\vert = N$ and existing a vertex $u\in\epsilon_i$. Then this form can be written as
\begin{align*}
\lambda&=\{\{u,x_i\}\}_{i=1}^N\\
&=\{(u,x_i),~(x_i,u)\}_{i=1}^N\\
\end{align*}
It implies
\begin{equation}\label{eqc2}
\lambda=\{ \{u\}\times A,~ A\times \{u\}\},\qquad A=\{x_i\}_{ i=1}^N .
\end{equation}
Hence, the set $\lambda$ is a Cartesian product set.\\
\end{proof}

\begin{definition}
As the described in Lemma\ref{cg2}, we call the set $\{u\}$ in form\eqref{eqc2} as root, denote by $R(\lambda)$. Similarly call the set $A$ as leaf set, use $L(\lambda)$ denote. We reserve the subscript of $\lambda$ equals to roots'.
\end{definition}

\begin{definition}
Let $\Lambda_p$ be a subset of collection $\Lambda$ of  edge relation subgraphs. Consider two components $ \lambda_i, \lambda_j\in\Lambda_p$. If there is a partition approach on set $\Lambda_p$, which for a pair $v_i\epsilon v_j$ such that $v_i\epsilon v_j\in\lambda_i\setminus \lambda_j$, or similar to contrary. We call this partition \emph{edge relation subgraph partition}.\\
\end{definition}

\begin{definition}
Given a no-empty $G_{\theta}=(V,~\epsilon)$. Let $A_q$ be a sequence of vertices on set $V$ and $A_q=V$. Consider two elements $x_i,x_j\in A_q$ with $i<j$.  There is a sequence $\Lambda_q$ of edge relation subgraphs with $\Lambda_q\subseteq \Lambda_p$. For two components $\lambda_s,\lambda_t\in\Lambda_q$. If for $x_i\in R(\lambda_s)$ and $x_j\in R(\lambda_t)$ such that $s<t$, then we call sequence $\Lambda_q$ \emph{ordered partition of edge relation subgraph}, abbr. by OPERS, call set $A_q$ \emph{ordered roots set}. \\
\end{definition}

\begin{lemma}\label{cg3}
Given a no-empty $G_{\theta}=(V,\epsilon)$. Let $\Lambda_q$ be an OPERS on set $\epsilon$ and $A_q$ be an ordered roots set of set $\Lambda_q$. Consider two elements $x_i,x_j\in A_q$ and $\{x_i,x_j\}\in\epsilon$ with $i<j$. If there are two components $ \lambda_s,\lambda_t\in\Lambda_q$ with $x_i\in R(\lambda_s)$ and $x_j\in R(\lambda_t)$, then $x_i\epsilon x_j \in \lambda_s\setminus\lambda_t$.
\end{lemma}

\begin{proof}
Given a no-empty $G_{\theta}=(V,\epsilon)$. There is an OPERS $\Lambda_q$ on set $\epsilon$ and $A_q$ be an ordered roots set of set $\Lambda_q$. Consider two elements $x_i,x_j\in A_q$ with $i<j$. We set two components $ \lambda_s,\lambda_t\in\Lambda_q$ with $x_i\in R(\lambda_s)$ and $x_j\in R(\lambda_t)$. If pair $\{x_i,x_j\}\in\epsilon$. With the definition of edge relation subgraph, observes there may be $x_i\epsilon x_j \in\lambda_s$ or $ x_i\epsilon x_j \in\lambda_t$. Further, if introduce pair $x_i\epsilon x_j$ to subgraph $\lambda_s$, then $ x_i\epsilon x_j\notin\lambda_t$ , similar to contrary.%

Hence, while in case $i<j$, the $i$th position on sequence $A_q$ is fore than $j$th, then pair $x_i\epsilon x_j$ can be partitioned into $\lambda_s$ fore than $\lambda_t$, i.e. there is $x_i\epsilon x_j\in \lambda_s\setminus\lambda_t$ in set $\Lambda_q$.\\
\end{proof}

\begin{theorem}\label{cg4}
Let $\Lambda_q$ be an OPERS on set $\epsilon$, then set $\Lambda_q$ is a partition of set $\epsilon$.
\end{theorem}

\begin{proof}
Let $\Lambda_q$ be an OPERS on set $\epsilon$. Consider two components $\lambda_i,\lambda_j\in\Lambda_q$ with $i\neq j$ such that $\lambda_i\neq \lambda_j$. Further, let $v_i\in R(\lambda_i)$ and $ v_j\in R(\lambda_j)$. If pair $\{v_i,v_j\}\nsubseteq\epsilon$, then $v_i\notin L(\lambda_j)$ and $ v_j\notin L(\lambda_i)$ with Lemma\ref{cg2}. Thus there is $\lambda_i\cap\lambda_j=\varnothing$. When $\{v_i,v_j\}\in\epsilon$, we have $v_i\epsilon v_j\in\lambda_i$ or $v_i\epsilon v_j\in\lambda_j$. We set the $i$th position is fore than $j$th ones on sequence $\Lambda_q$. We have $v_i\epsilon v_j\in \lambda_i\setminus\lambda_j$ as Lemma\ref{cg3}. Hence $\lambda_i\cap\lambda_j=\varnothing$ similarly. %

Consider pair $\{v_i,v_j\}\in\epsilon$ such that $v_i\epsilon v_j\in\lambda_i$ or $v_i\epsilon v_j\in\lambda_j$, so that for each component $\lambda_k\in \Lambda_q$, there is $\lambda_k\neq\varnothing$. Furthermore, for each pair $(u,v)\in\epsilon$ we have $(u,v)\in\Lambda_q$, then there is $\epsilon=\Lambda_q$. To sum up above, set $\Lambda_q$ is the partition of set $\epsilon$. Consider the case $v_i\epsilon v_j\in\lambda_i$ or $v_i\epsilon v_j\in\lambda_j$, we understand set $\Lambda_q$ is no an equivalent class of set $\epsilon$.\\
\end{proof}

\begin{theorem}\label{cg5}
There is an OPERS $\Lambda_q$ on an instance $G_{\theta}=(V,\epsilon)$, then $\vert \Lambda_q\vert < n$
\end{theorem}

\begin{proof}
Given a no-empty $G_{\theta}=(V,\epsilon)$. There is an OPERS $\Lambda_q$ on set $\epsilon$, which roots set is $A_q$. Because $A_q$ is a sequence on set $V$, therefore $\vert A_q\vert =n$. For each element $u\in A_q$, if there may be $u\in R(\lambda_k)$ and $\lambda_k\in \Lambda_q$, we understand there possibly is $\vert \Lambda_q\vert = n$.%

Further, we can set there are two sub-sequences $\alpha_1,\alpha_2$ on set $A_q$ and $A_q=(\alpha_1,\alpha_2)$. Consider such component $\lambda_s\in\Lambda_q$. If there is $L(\lambda_s)\setminus\alpha_1=\varnothing$ and $ R(\lambda_s)\subseteq \alpha_2$, then it implies that each pair in subgraph $\lambda_s$ may be introduced into other subgraphs with Lemma\ref{cg3}. Hence, no such component exists in set $\Lambda_q$, i.e. $\vert\Lambda_q \vert < n$.%

Assume to $\vert \Lambda_q\vert = n$. Then there is the end component $\lambda_n\neq\varnothing$ on sequence $\Lambda_q$. With Lemma\ref{cg3}, at least have a vertex $x_t\in L(\lambda_n)$ and $x_t\in A_q$ with $ t>n$, i.e. $\vert A_q\vert >n$ is a contradiction.\\
\end{proof}

\begin{lemma}\label{cg6}
Given a no-empty $G_{\theta}=(V,~\epsilon)$. If $\Lambda'$ is the collection of OPERSs, then $\vert\Lambda' \vert\leq n!$ 
\end{lemma}

\begin{proof}
Given a no-empty $G_{\theta}=(V,~\epsilon)$. Let $A_q$ be a sequence on set $V$, then there are $n!$ arrangements of vertices in set $A_q$ with $A_q=V$. We let $\Lambda'$ be a collection of OPERSs. For a component $\Lambda_q\in\Lambda'$, we can set sequence $A_q$ is the roots set of $\Lambda_q$. Hence there may be $n!$ arrangements of subgraphs in set $\Lambda_q$.%

Consider instance $G_{\theta}\in G_{\eta}$, on which there is each pair $u,v\in V$ such that $\{u,v\}\in\epsilon$ with $u\neq v$. Then for each component $\lambda_i\in\Lambda_q$, we may have $\vert L(\lambda_i) \vert \leq n-1$. Further, for component $\lambda_{n-1}\in\Lambda_q$, there may be $\vert L(\lambda_{n-1}) \vert =1$, and no $\lambda_{n}$ on sequence $\Lambda_q$ as Lemma\ref{cg3}. Hence, for $n$th position of sequence $A$, there possibly are $n$ vertices on it, and $(n-1)!$ arrangements for $n-1$ vertices on $n-1$ positions on sequence $A_q$. For a pair $x_i,x_j\in A_q$ such that $\{x_i,x_j\}\in\epsilon$, then the relation of position between $i$th and $j$th can change two components $\lambda_i, \lambda_j\in\Lambda_q$ with $x_i\in R(\lambda_i)$ and $x_j\in R(\lambda_j)$. Namely, there is $n!$ possibilities of arrangements of subgraphs in set $\Lambda_q$, i.e. $\vert\Lambda' \vert= n!$%

If $G_{\theta}\notin G_{\eta}$, at least some cardinalities of leaf sets are less than $n-1$. Obviously, at least we can find two vertices $x_i,x_j\in A$ such that having $\lambda_i,\lambda_j\notin\Lambda_q$ with Lemma\ref{cg5}. Then there at most are $P_n^{n-2}$ possibilities of arrangement of subgraphs in set $\Lambda_q$. Hence $\vert\Lambda' \vert< n!$.%

Consequently, we prove that set $\Lambda_{q}$ is not an equivalent class on set $\epsilon$ again.\\
\end{proof}

\begin{descussion}
Author constructed such logic structure $\epsilon\subseteq\tau\subseteq V^2$. It is necessary to show the difference between edge relation and traversal relation. With traversal relation, these pairs $u\tau v,~v\tau u$ are not equivalent. You can not decide that a visiting from A city to B city equals to it from B city to A city. But in Geometry space, segment AB equals to BA. If we define \emph{edge path} like above, then a minimum edge cycle can be $\mathring{\pi}_{\epsilon}=\{\{A,~B\},\{B,~C\},\{C,~A\}\}$. Obviously, this form is a triangle, which is the minimum polygon. It is just to the reason that there is the symmetry property on edge. %

As the described in Lemma\ref{cg6}, the partition of edge relation subgraph is not an equivalent class on $\epsilon$. Such that if we use the edge relation for graph traversal, then we would meet a big trouble to enumerate much more possibilities of adjacency. That is why author defined simple graph in this section. The graph coloring problem must be discussed underlaying this logic structure with symmetry property, but not for graph traversal.\\
\end{descussion}

\subsection{Interval Vertices Set}
Here we will define a special class: interval vertices set. Indeed, there are logic relations between two interval vertices as the coloring relation\eqref{color1}. We will explore these features of those vertices in this subsection.

\begin{definition}
Consider a pair $u,v$ with $u\neq v$ on instance $G_{\theta}=(V,\epsilon)$. If pair $\{u,v\}\in(V^2\setminus\epsilon)$, then we call them \emph{interval vertices}. The collection of these interval vertices we denote by $\bar{\epsilon}$ and call it \emph{interval vertices set}.%

Let $\delta$ be a subset of set $\bar{\epsilon}$. For each pair $u, v\in\delta$ such that there is $u,v\in \bar{\epsilon}$, we call set $\delta$ \emph{completed interval vertices set}, abbr. by CIVS.
\end{definition}

\begin{lemma}\label{cg7_1}
Given a no-empty CIVS $\delta$, then $\vert \delta\vert\geq 2$.
\end{lemma}

\begin{proof}
Given a no-empty CIVS $\delta$. If $\vert \delta\vert=1$, then we let $u\in\delta$. It is a contradiction for pair $\{u,u\}\in\bar{\epsilon}$. Hence there is no such CIVS, $\vert \delta\vert\geq 2$.\\
\end{proof}

\begin{theorem}\label{cg7}
Let $\Lambda_p$ be a partition of edge relation subgraph on set $\epsilon$. Consider two components $\lambda_i,\lambda_j\in\Lambda_p$ with $\lambda_i\neq\lambda_j$. If their is such a case $R(\lambda_i)\cap L(\lambda_j)=\varnothing$ and $R(\lambda_j)\cap L(\lambda_i)=\varnothing$, then $R(\lambda_i),R(\lambda_j)\in\bar{\epsilon}$.
\end{theorem}

\begin{proof}
Let $\Lambda_p$ be a partition of edge relation subgraph on instance set $\epsilon$. There are two components $\lambda_i,\lambda_j\in\Lambda_p$ with $\lambda_i\neq\lambda_j$. Moreover Let vertices $x_i\in R(\lambda_i)$ and $ x_j\in R(\lambda_j)$, we have $x_i\neq x_j$. We set 
\begin{align*}
I(i,j)=&R(\lambda_i)\cap L(\lambda_j);\quad I(j,i)=R(\lambda_j)\cap L(\lambda_i).
\end{align*}
Consider $ I(i,j)=\varnothing$ and $ I(j,i)=\varnothing$. As Lemma\ref{cg2}, there are  $x_i\notin L(\lambda_j)$ and $ x_j\notin L(\lambda_i)$. As the definition of edge relation subgraph, there is no such pair $\{x_i,x_j\}\in \epsilon$. Hence, we have $\{x_i,x_j\}\in\bar{\epsilon}$.%

Assume to $\{x_i,x_j\}\in\bar{\epsilon}$. Then  $\{x_i,x_j\}\notin\epsilon$, it implies $ x_j\notin\lambda_i$ and $x_i\notin\lambda_j$. Hence we have $I(i,j)=\varnothing$ and $ I(j,i)=\varnothing$. We can understand that the given condition is the necessary and sufficient condition in this Theorem.\\
\end{proof}

\begin{definition}
Let $\Lambda_q$ be an OPERS on set $\epsilon$ and $A_q$ be the roots set of set $\Lambda_q$. Consider $\Lambda_e$ is a subset of $A_q$. Given each vertex $v_i\in\Lambda_e$, there is no such edge relation subgraph $\lambda_i$ for $v_i\in R(\lambda_i)$ and $\lambda_i\in \Lambda_q$. We call set $\Lambda_e$ \emph{empty subgraph set}.
\end{definition}

\begin{lemma}\label{cg8}
Let $\Lambda_e$ be an empty subgraph set and $\vert \Lambda_e\vert\geq 2$ on an instance, then $\Lambda_e\subseteq\delta$.
\end{lemma}

\begin{proof}
Let $\Lambda_e$ be an empty subgraph set and $\vert \Lambda_e\vert=N$. If $N=1$, then it $\Lambda_e\subseteq\delta$ does not hold with Lemma\ref{cg7_1}. Consider $N\geq 2$ and two vertices $v_i,v_j\in \Lambda_e$ with $v_i\neq v_j$. With definition of empty subgraph, there are no such components $\lambda_i, \lambda_j\in\Lambda_p$ for $v_i\in R(\lambda_i)$ and $v_j\in R(\lambda_j)$. Then there is no pair $\{v_i,v_j\}\in\epsilon$ with Theorem\ref{cg7}. Hence, have $\Lambda_e\subseteq \delta$ with definition of CIVS.%

Assume that there is a pair $v_i,v_j\in\Lambda_e$ and $\{v_i,v_j\}\notin \bar{\epsilon}$. It is obvious for $\{v_i,v_j\}\in\epsilon$ such that the pair can not be partitioned into other subgraph besides $\lambda_i$ or $\lambda_j$. As the descirbed in Theorem\ref{cg7}, we have that subgraph $\lambda_i$ or $\lambda_j$ is no empty, a contradiction. \\
\end{proof}

\begin{definition}
Let $\Delta$ be a subset of a interval vertices set $\bar{\epsilon}$. Given each component $h_i\in \Delta$ having $h_i\subseteq\delta$. For each pair $h_i,h_j\in\Delta$ with $h_i\neq h_j$, then there is $(h_i\cup h_j)\nsubseteq\delta$. We call set $\Delta$ \emph{minimum completed interval vertices set}, abbr. by MCIVS.\\
\end{definition}

\begin{lemma}\label{cg9_1}
Let $\Delta$ be MCIVS on an instance set $\bar{\epsilon}$, then set $\Delta$ is the partition of set $\bar{\epsilon}$.
\end{lemma}

\begin{proof}
Let $\Delta$ be MCIVS on a no-empty instance set $\bar{\epsilon}$. We set $R=\bar{\epsilon}\setminus \Delta$. Consider $R\neq\varnothing$. For $\vert R\vert=1$, then set vertex $u\in R$. If there is at least component $\delta_k\in \Delta$ such that $(\delta_k\cup \{u\})\subseteq\delta$, then we have $u\in \delta_k$ and $R=\varnothing$. Given each vertex $v\in\Delta$. If there is pair $\{u,v\}\in\epsilon$, then $u\notin \bar{\epsilon}$ with definition of interval vertices set, we have $R=\varnothing$. If only there is a vertex $v\in\Delta$ such that pair $\{u,v\}\in\bar{\epsilon}$, then we have a component $\{u,v\}\subseteq\Delta$. Hence, observe $\vert R\vert=1$, it does not holds.%

Consider $\vert R \vert>1$. Summarizing the case of $\vert R\vert=1$, for each vertex $u\in R$ such that we may have $u\in\Delta$ or $u\notin\bar{\epsilon}$. Hence, we can understand $\bar{\epsilon}\setminus \Delta=\varnothing$.%

Assume that for each pair $\delta_i, \delta_j\in\Delta$ with $\delta_i\neq \delta_j$ such that $\delta_i\cap\delta_j\neq\varnothing$. Consider $\vert\delta_i\vert= 2$ and $\vert\delta_j\vert=2$, we set $v\in(\delta_i\cap\delta_j)$. If we remove $v$ from one of components, we may have the case of $\vert R\vert=1$ above. If $\vert\delta_i\vert>2$ and $\vert\delta_j\vert> 2$, then vertex $v$ in $\delta_i$ or $\delta_j$ can not affect $\delta_i\subseteq\delta$ or $\delta_j\subseteq\delta$. Because we define the set $\Delta$ as no a multi-set, therefore there may be $\delta_i\cap\delta_j=\varnothing$. Consequently, we understand that the MCIVS impossibly is an equivalent class on set $\bar{\epsilon}$.%

As Lemma\ref{cg7_1}, we understand that given a component $\delta_i\in\Delta$, there is $\vert \delta_i\vert\geq 2$. Hence, the MCIVS $\Delta$ is the partition of interval vertices set.\\
\end{proof}

\begin{theorem}\label{cg9}
Given a no-empty $G_{\theta}=(V,\epsilon)$.  Let $\Delta$ be a MCIVS on set $V$. If $G_{\theta}\in G_{\eta}$, then $\Delta=\varnothing$.
\end{theorem}

\begin{proof}
Given a no-empty $G_{\theta}=(V,\epsilon)$.  There is a MCIVS $\Delta$ on set $V$. Consider $n=1$. There is $G_{\theta}\in G_{\eta}$ with definition of completed regular graph. As Lemma\ref{cg7_1}, there is no interval vertices set on instance. Hence, $\Delta=\varnothing$ holds for given instance is an isolated node.%

Consider $n>1$. For each pair $u,v\in V$ with $u\neq v$ such that there is $\{u,v\}\in\epsilon$. Hence there is no any pair in set $\bar{\epsilon}$, i.e. set $\bar{\epsilon}$ is empty. Then $\Delta=\varnothing$ holds.%

If there is $\Delta=\varnothing$ on instance. As Lemma\ref{cg9_1}, we have $\Delta=\bar{\epsilon}$ such that no any pair in set $\bar{\epsilon}$. Hence, there is $G_{\theta}\in G_{\eta}$. We understand that the condition of $G_{\theta}\in G_{\eta}$ is the necessary and sufficient condition.\\
\end{proof}

\begin{theorem}\label{cg10}
Let $S$ is the collection of unit subgraphs on an instance $G_{\theta}=(V,\epsilon)$. Given each component $s_i\in S$ such that $\vert L(s_i)\vert\equiv m$. If  $ 2\leq m< n-1$, then there is a MCIVS $\Delta$ on set $\bar{\epsilon}$ with $2\leq\vert \Delta\vert\leq m$.
\end{theorem}

\begin{proof}
Let $S$ is the collection of unit subgraphs on an instance $G_{\theta}=(V,\epsilon)$. For each component $s_i\in S$ we have $\vert L(s_i)\vert\equiv m$. We set $ 2\leq m< n-1$. Furthermore, let $\bar{\epsilon}$ be an interval vertices set on set $V$.%

As definition of completed regular graph, the graph $G_{\theta}$ is not a completed regular graph with $m<n-1$. Then the set $\bar{\epsilon}$ is no empty such that the MCIVS $\Delta$ on set $\bar{\epsilon}$ is no empty with Theorem\ref{cg9}. It is obviously that if $n=1$ or $n=2$, then $G_{\theta}\in G_{\eta}$ such that these cases can not satisfy the condition of $m<n-1$. Hence, we have to consider the case $n\geq 3$.%

Consider each unit subgraph $s_i$ with vertex $v_i\in R(s_i)$. Set there is a cutting graph $G'=s_i$ and let $G'=(V',~\epsilon')$. We assume that for each pair $v_j,v_k\in L(s_i) $ such that $\{v_j,v_k\}\in\epsilon$. With definition of unit subgraph, we have those pairs $v_i\epsilon v_j,v_i\epsilon v_k$ on $G'$, then for each pair $u,v\in G'$ such that $\{u,v\}\in\epsilon'$, we may have $G'\in G_{\eta}$ and $\vert V'\vert =m+1$. If $G'=G_{\theta}$, then this assumption contradicts the given conditions. In case $G'\subset G_{\theta}$, then the instance $G_{\theta}$ is no connected. Hence $G'\notin G_{\eta}$ and $V'\cap\bar{\epsilon}\neq\varnothing$.%

We can decide that for these leaves on subgraph $s_i$, there are at least two vertices $u,v\in\bar{\epsilon}$, thus all leaves at most have to be partitioned into $m-1$ CIVSs. Observe that $m$ CIVSs may satisfy the partition for those vertices on subgraph $s_i$. If $L(s_i)\setminus \Delta=\varnothing$ and $\vert \Delta\vert=m$, we understand there is at least component $\delta_j\in\Delta$ to be introduced with root $v_i$, similar to case of $L(s_i)\setminus \Delta\neq\varnothing$. While $L(s_i)\cap\Delta=\varnothing$, then it is certainly that each component $\delta_j\in\Delta$ can be introduced with the root $v_i$. Hence, $\vert \Delta\vert\leq m$.%

Consider there is a partition of graph $\Sigma$ on instance $G_{\theta}$. If there is $\vert \Sigma\vert \geq 2$ and no any traversal relation on $V$ direction in each region, then for each pair $u,v$ in a region such that have $u,v\in\bar{\epsilon}$. Then we have for each $\sigma_i\in \Sigma$, if $\vert \sigma_i\vert\geq 2$ then there may be $ \sigma_i\subseteq \delta$ with definition of CIVS. With Lemma\ref{gp3} and Lemma\ref{gp7}, we understand there is no traversal relation between $\sigma_i$ and $\sigma_{i+2}$. Then we have that
\begin{align*}
\delta_1&=\bigcup_{k=1}^N\sigma_{2k}\quad \delta_2=\bigcup_{k=0}^N\sigma_{2k+1},\quad (N\leq \lfloor\vert \Sigma\vert /2\rfloor).\\
\Delta&=\{\delta_1,~\delta_2\}.
\end{align*}
Hence $\vert\Delta\vert=2$, i.e. it $2\leq\vert \Delta\vert\leq m$ holds. We understand that we can utilize graph partition to find out CIVS on an instance.\\
\end{proof}

\begin{descussion}
If we view the interval as a relation between two vertices, then this relation can possess transitivity and symmetry. As for the relation of $=$ or $\neq$  similarly possesses these properties, we can understand it is no an occasional case of existing the minimum chromatic set on an instance. Theorem\ref{cg9_1} shows that the MCIVS  is no an equivalent class such that there exist more types of partition on set $\bar{\epsilon}$, which is easy to confuse us. Hence, author will only show the speed-up approximation algorithms, but they have a nice precision.\\
\end{descussion}

\subsection{Minimum Chromatic Set}
\begin{definition}
Let $V'$ be a subset of set $V$ and $\vert V'\vert \geq 2$. For labeling all vertices in set $ V'$, if there exist a relation having three properties of reflexive, symmetry and transitivity among those coloring values, then we say there is a relation of \emph{equivalence-color} on set $V'$, denote by $\kappa$. Reserve the abbr. $V'\subseteq\kappa$ to represent this concept. 
\end{definition}

\begin{definition}
Given a no-empty $G_{\theta}=(V,~\epsilon)$. Let $V'$ be a subset of set $V$ and $\vert V'\vert \geq 2$. For labeling all vertices in set $ V'$, if there exist a relation having two properties of symmetry and transitivity among those coloring values, then we call this relation \emph{inequivalence-color} on set $V'$, denote by $\bar{\kappa}$. Reserve the abbr. $V'\subseteq\bar{\kappa}$ to represent this concept.
\end{definition}

\begin{theorem}\label{MC1}
Let $V'$ be a collection of vertices and $\vert V'\vert \geq 2$. If $V'\subseteq\delta$,  the  interval relation $\bar{\epsilon}$ possesses these properties of symmetry and transitive on set $V' $, similar to $V'\subseteq G_{\eta}$.
\end{theorem}

\begin{proof}
Given an instance $G_{\theta}=(V,~\epsilon)$. Let $V'$ be a subset of set $V$ and $\vert V'\vert \geq 2$. Consider the case $V'\subseteq\delta$ without any edge relation on it. For each pair $u,v\in V'$ such that the pair $u,v\in\delta$. As the definition of interval vertices, there is pair $\{u,v\}\in V^2\setminus\epsilon$. Hence, we have these pairs $(u,v),(v,u)\subseteq\bar{\epsilon}$, i.e. there is the property of symmetry in interval relation.%

Consider $\vert V'\vert>2$ and there are three vertices $u,v,t\in V'$. With the definition of CIVS, these vertices $u,v,t$ are interval to each other. Then for pairs $\{u,v\},\{v,t\}\subseteq\bar{\epsilon}$, we similarly can have $\{u,t\}\subseteq\bar{\epsilon}$. Hence, there exists the property of transitive on set $V'$.%

Similarly we can use the same fashion to prove there existing two properties of symmetry and transitivity on set $V'$ with $V'\subseteq G_{\eta}$.\\
\end{proof}

\begin{theorem}\label{cg11}
There is a coloring relation $g$ on an instance set $V$. Let $V'$ be a subset of set $V$. If $V'\subseteq\delta$, then there may be $g_{\bar{\epsilon}}:V'\rightarrow \kappa$.
\end{theorem}

\begin{proof}
Given an instance $G_{\theta}=(V,\epsilon)$. Let $V'$ be a subset of set $V$ and $V'\subseteq\delta$. We set there is a coloring relation $g$ on set $V$. Then for each pair $u,v\in V'$ such that their chromatic values have $C_u=C_v$ or $C_u\neq C_v$ with definition of coloring relation\eqref{color1}.%

Consider there are three vertices $u,v,t\in V'$ and we choose the relation of unequal $\neq$ on set $V'$. It is obvious that the case $C_u=C_t$ may satisfy the condition of $C_u\neq C_v$ and $C_v\neq C_t$. Hence it may contradict the property of transitive on CIVS $\delta$.%

If we choose the relation of equal = on set $V'$. It is naturally that there is $C_u=C_t$ with $C_u=C_v$ and $C_v= C_t$, satisfies the property of transitive on set $V'$. And it similarly satisfies existing two properties of symmetry and reflexivity on set $V'$. Hence there is $g_{\bar{\epsilon}}:V'\rightarrow \kappa$.%

Let $A$ be the subset of set $V$ and $A\subseteq\kappa$. Consider each pair $u,v\in A$ with $u\neq v$. For $C_u=C_v$ such that there is pair $u,v\in\bar{\epsilon}$ with the definition of coloring relation\eqref{color1}. Hence $A\subseteq \delta$. It $g_{\bar{\epsilon}}:V'\rightarrow \kappa$ is the necessary condition in this Theorem.\\
\end{proof}

\begin{theorem}\label{cg12}
There is a coloring relation $g$ on an instance set $V$. Let $V'$ be a subset of set $V$. If $V'\subseteq G_{\eta}$, then there is  a coloring relation $g$ on set $V'$ such that $g_{\epsilon}:V'\rightarrow  \bar{\kappa}$.
\end{theorem}

\begin{proof}
There is a coloring relation $g$ on an instance set $V$. Let $V'$ be a subset of set $V$. We set $V'\subseteq G_{\eta}$. Then for each pair $u,v\in V'$ such that these is $C_u\neq C_v$ with definition of coloring relation\eqref{color1}. Consider there are three vertices $u,v,t\in V'$. Then we have $C_u\neq C_v\neq C_t$ and $C_u\neq C_t$. It is certainly that the case can satisfy those properties of symmetry and transitivity.%

Because of no existing $C_u\neq C_u$, then there is the relation of inequivalence-color on set $V'$, therefore, we have that $g_{\epsilon}:V'\rightarrow \bar{\kappa}$. Consider there is a set $A$ with $A\subseteq\bar{\kappa}$, for pair $u,v\in A$ and $u\neq v$ with $u,v\in\bar{\epsilon}$, there may be a case $C_u\neq C_v$ with definition of coloring relation $g$. Therefore $g_{\epsilon}:V'\rightarrow \bar{\kappa}$ is no necessary condition in this Theorem.\\
\end{proof}

\begin{theorem}\label{cg14}
Let $V'$ be a subset of set $V$ on an instance. There is an approach to use chromatic values set $C$ to label each vertex in set $V'$ as coloring relation $g$. If $V'=\Delta$, then $\vert C\vert \geq\vert \Delta\vert$
\end{theorem}

\begin{proof}
Let $V'$ be a subset of set $V$ on an instance. There is an approach for labeling those vertices in set $V'$ as coloring relation $g$. We set $C$ is the chromatic values set, moreover let $V'=\Delta$ and $ C =c_1,c_2,\cdots,c_k$ for $i\neq j$ such that $ c_i\neq c_j$. %

We set the least colors is $k$ for labeling those vertices in set $V'$. Then there are $k$ components in set $V'$. For each component $\alpha_i\in V'$, it is certainly for each vertex $ u\in \alpha_i$ such that $C_u=c_i$ for $c_i\in C$. Obviously, if $\vert \alpha_i\vert \geq 2$, then with Theorem\ref{cg11}, $\alpha_i\subseteq\kappa$ having $\alpha_i\subseteq\delta$. Then for $V'=\{\alpha_i\}_{i=1}^k$ such that $V'\subseteq\Delta$. %

Assume that there is $\alpha_i\nsubseteq \delta$. Then there is a pair $u,v\in\alpha_i$ with $u\neq v$. If $\{u,v\}\in \epsilon$, it is obvious a contradiction for this pair having $C_u=C_v$ to the definition of coloring relation $g$. Hence $\vert C\vert =\vert \Delta\vert$. Because of maybe existing the case $u,v\in \bar{\epsilon}$ and $C_u\neq C_v$, therefore $\vert C\vert \geq\vert \Delta\vert$.\\
\end{proof}

\begin{definition}
Let $\Delta$ is a MCIVS on an instance $G_{\theta}=(V,\epsilon)$. We call the complement of $V\setminus\Delta$ \emph{edge association set}, denote by $\overline{V}$.
\end{definition}

\begin{theorem}\label{cg10}
Let $\Delta$ is a MCIVS on an instance $G_{\theta}=(V,\epsilon)$. There is an edge association set $\overline{V}=V\setminus\Delta$. Set existing a relevant cutting graph $\overline{G}=(\overline{V},\epsilon')$(for $\epsilon'\subseteq\epsilon$). If $\overline{V}\neq\varnothing$, then $\overline{G}\in G_{\eta}$.
\end{theorem}

\begin{proof}
Let $\Delta$ is a MCIVS on an instance $G_{\theta}=(V,\epsilon)$. There is an edge association set $\overline{V}=V\setminus\Delta$ with a relevant cutting graph $\overline{G}=(\overline{V},\epsilon')$ with $\epsilon'\subseteq\epsilon$.

Consider $\Delta=\varnothing$, then there is $\overline{G}=G_{\theta}$. For each pair $u,v\in V$ such that $u,v\notin\bar{\epsilon}$. We can have $\overline{G} \in G_{\eta}$ with Theorem\ref{cg9}. %

If $\vert \overline{V}\vert = 1$. There is $\overline{G} \in G_{\eta}$ as the definition of completed regular graph. Consequently, there is $\overline{G}\in G_{\eta}$ if $\Delta=\varnothing$ or $\vert \overline{V}\vert = 1$.%

Consider $\Delta\neq\varnothing$ and $\vert \overline{V}\vert >1$. There is $\overline{V}\cap \Delta=\varnothing$ as the given condition of $\overline{V}=V\setminus\Delta$. Thus for each pair $u,v\in\overline{V}$, there is $u,v\notin\Delta$. Because of $\Delta=\bar{\epsilon}$ with Lemma\ref{cg9_1}, we can understand the pair $ \{u,v\}\in\epsilon$. Hence, $\overline{G} \in G_{\eta}$.%

Assume to $\overline{G}\nsubseteq G_{\eta}$. Then there is at least a pair $u,v\in\overline{G}$ such that $u,v\in\bar{\epsilon}$. As described in Lemma\ref{cg9_1}, the set $\Delta$ is the partition of set $\bar{\epsilon}$, then there is $\Delta\cap\overline{V}\neq\varnothing$, a contradiction to $\overline{V}=V\setminus\Delta$.%

If there is a no-empty subset $A\subset \overline{V}'$, then there is $V_h=V\setminus A$. It $V_h\subseteq\Delta$ obviously may not hold. Hence, it $\overline{G}\in G_{\eta}$ is not the necessary condition in this Theorem.\\
\end{proof}

\textbf{Proposition. }Consider $\vert V\vert \geq 3$ and for each unit subgraph $s_i$ such that $L(s_i)\equiv m $. If $2\leq m< n-1$, then there is $0\leq\vert\overline{V}\vert\leq 1$ for $\overline{V}=V\setminus\Delta$.%

For example, instance is a cycle. If the number of $\vert V\vert$ is odd, then $\vert\overline{V}\vert=1$. Otherwise $\vert\overline{V}\vert=0$. Is it true or false, how to prove?\\

\begin{theorem}\label{cg15}
Let $\Delta$ be a MCIVS and $\overline{V}$ be the edge association set. Consider there is an approach to use chromatic values set $C$ to label each vertex on graph $G_{\theta}=(V,\epsilon)$ as coloring relation $g$. If $\vert C\vert=K$, then $K\geq\vert \Delta\vert + \vert\overline{V}\vert$.
\end{theorem}

\begin{proof}
Given a no-empty $G_{\theta}=(V,~\epsilon)$. Let $\Delta$ be a MCIVS on graph $G_{\theta}$ and $\overline{V}=V\setminus\Delta$. There is an approach to use chromatic values set $C$ to label each vertex in set $V$ as coloring relation $g$. Moreover let $\vert C\vert=K$. As the definition of edge association set, we have $V=\Delta\cup\overline{V}$. As Theorem\ref{cg14}, there is $\Delta\subseteq\kappa$. As Theorem\ref{cg12} and Theorem\ref{cg10}, there is $\overline{V}\subseteq\bar{\kappa}$. Hence the number of colors to label each vertex in set $V$ is the sum of $\vert \Delta\vert + \vert \overline{V}\vert$, i.e. $K\geq\vert \Delta\vert + \vert\overline{V}\vert$.%

We assume there is a minimum chromatic set $C'$, in which there are $m$ colors with $m<\vert \Delta\vert + \vert\overline{V}\vert$, satisfy labeling each vertex in set $V$ as coloring relation $g$. Then there is $m$ components in set $V$. For each component $\alpha_j\in V$ with each vertex $u\in \alpha_j$ such that $C_u=c_j $ and $ c_j\in C'$. When $\vert \alpha_j\vert>1$, then there is $\alpha_j\subseteq\delta$ as Theorem\ref{cg14}. We have $V\subseteq \Delta$ such that $V\setminus\Delta=\varnothing$, having $\overline{V}=\varnothing$ and $m=\vert \Delta\vert$. Namely, the inequality above holds similarly. Otherwise, when $\vert \alpha_j\vert=1$, if there exists a pair $u,v\in V$ such that $u,v\in\bar{\epsilon}$, it is a contradiction to assumption of $m$ being minimum color value. Hence, $\Delta=\varnothing$ and $m=\vert\overline{V}\vert$, the inequality above holds too.\\%

Consider there are two classes $A_1, A_2\in V$. For each components $\alpha_i\in A_1$ such that there is $\vert\alpha_i\vert=1$. Otherwise, $\vert\alpha_j\vert>1$ has $\alpha_j\in A_2$. First we understand class $A_2\subseteq\Delta$ as Theorem\ref{cg11}. For class $A_1$, we have proved the case $A_1\in \overline{V}$ above by $m$ being minimum color value. Hence, $m=\vert \Delta\vert + \vert\overline{V}\vert$, i.e. inequality $K\geq\vert \Delta\vert + \vert\overline{V}\vert$ holds.\\

\end{proof}

\begin{descussion}
Finally, Theorem\ref{cg15} shows that looking for the minimum chromatic value indeed is to partition each vertex in set $V$ to two relation classes, equivalence-color and inequivalence-color. Then the algorithms can be two types. One is looking for the MCIVS on graph $G$. Other is utilizing the edge relation subgraph to record the info. of the inequivalence-color relation for coloring each root.\\
\end{descussion}

\subsection{Coloring Algorithm and K Value}
Author gives two approximation algorithms, they respectively are \emph{Based On Graph Partition Coloring}( abbr. BOGPC ) and\emph{ Based On Edge Relation Coloring}( abbr. BOERC ), which both choose those vertices by random.%

To BOGPC, we understand these interval vertices nearest to the start-nodes are in the region $\sigma_3$ with Lemma\ref{gp3}. Then the approach can iteratively use the method of graph partition to enumerate those vertices and merge them with start-nodes, until can not partition out the region $\sigma_3$. Hence, we define the iterated graph partition function in following, denote by $\phi$\\

\begin{align*}
&\phi(X,G)\coloneq\left\{
\begin{array}{ll}
X\leftarrow\sigma_3,&\text{if  }\exists\sigma_3\in\Sigma \text{ and } \Sigma=V .\\
~&~\\
\text{undefine,}&\text{otherwise}.\\
\end{array}\right.\\
&~\\
&\phi^{m+1}=\phi\cdot\phi^{m}.
\end{align*}
~\\

Function will return a set $X$ as a CIVS. Then the BOGPC can iteratively cut the graph with set $X$ until the remainder vertices set belongs to edge association set $\overline{V}$ or empty relation subgraph set $\Lambda_e$. The pseudocode of algorithm is given as follow\\

\renewcommand\arraystretch{1.0}
\begin{longtable}{ p{120mm}}
\toprule  
\textbf{Algorithm 5: BOGPG}\\
\toprule   
\textbf{input} graph $G=s_1,s_2, \cdots, s_n; $\\
\qquad\quad set $H\leftarrow v$\\
\qquad\quad set $V \coloneq V\setminus H$\\
\qquad\quad set $R$\\
\textbf{output}  $R$\\
\textbf{00}\quad         \textbf{While}$(V\neq\varnothing)$\\
\textbf{01}\qquad            $S=H;~ H\coloneq\phi (H,G);$ \\
\textbf{02}\qquad 		$R(k)\leftarrow~ H$;\\       
\textbf{03}\qquad        \textbf{If}$(S\setminus H=\varnothing)$ \textbf{Than break}; \\ 
\textbf{04}\qquad 		$V\coloneq  V\setminus H$;\\
\textbf{05}\qquad 		$G\coloneq G\setminus (s_i\cup\beta_i )~\Longleftarrow \forall v_i\in H;$\\     
\textbf{06}\qquad         \textbf{If}$(V=\varnothing )$ \textbf{Than break};\\
\textbf{07}\qquad        $H\leftarrow (v\cup u)\Longleftarrow v\in\Lambda_e \wedge u=\textbf{random}(V);$\\
\textbf{08}\quad  \textbf{Output }$R=c_1, c_2,\cdots, c_k$(for $c_i=\{v_1,v_2,\cdots\}$)\\
\bottomrule
\end{longtable}

~\newline
\textbf{Complexity and Algorithms. }The core problem is the count of invoking graph partition function by program. Let $m$ be the cardinality of leaf set of each unit subgraph. As Theorem\ref{cg10}, in the worst-case the program can invoke function for $m+1$ times. For per-time, in worst-case pre-invoking equals to partition whole graph. And in worst-case, we let the $m$ is the maximum value in whole unit subgraph. Hence, the runtime complexity is $O(mn^3)$. The runtime complexity of cutting graph is $O(\tau^2)$(for $\tau=mn$).\\%

But you must understand that there exist possible edge relations on $V$ direction, which leads to the region $\sigma_3$ may not belong to a CIVS. Then it produces a problem how to introduce each vertex. Author uses the method of random selection, which is to randomly pick up a vertex $v_i$. If subgraph $s_i$ satisfied the condition of $ L(s_i)\cap X=\varnothing$, then root $v_i$ is valid. Hence this can affect the precision of looking for the minimum chromatic value $K$. Consequently, there is a problem that how much probabilistic can this approach find out the minimum chromatic value $K$. This problem will be discussed with BOERC in following. Now we need prove the precision of approach BOGPC is less than and equal to $m+1$ in following Lemma.\\%

\begin{lemma}
Let $m$ be the cardinality of a leaf set in unit subgraph. On an instance $G_{\theta}=(V,\epsilon)$, if coloring each vertex in set $V$ by approach BOGPC, then the chromatic value $ K \leq m + 1$.
\end{lemma}

\begin{proof}
Given a no-empty $G_{\theta}=(V,\epsilon)$. Let $m$ be the cardinality of a leaf set in unit subgraph. There is an approach BOGPC for coloring each vertex in set $V$. If $G_{\theta}\in G_{\eta}$, there is $m=n-1$ to each leaf set. Because for each pair in set $\epsilon$, therefore either graph or cutting graph is always partitioned to only two regions by function $\phi$ with definition of graph partition. Then function always returns a vertex in set $H$. Hence, The set $V$ should be partitioned for $n$ times, the chromatic value $K=m+1$.\\%

Consider $G_{\theta}\notin G_{\eta}$ and $m$ is a constant with $2\leq m<n-1$. Because the approach iteratively cuts the CIVSs $\delta$ from graph, therefore for a unit subgraph, there are at most $m$ times to cut its leaves from leaf set. Hence, the root should be introduced to the $(m+1)$th set. If the potential of that $(m+1)$th set greater than 1, then it is a CIVS with Lemma\ref{cg8}. When the potential is 1, then it is a completed regular graph. However, the chromatic value $K\leq m+1$. %

When the $m$ maybe difference to each other, in worst-case we can set $m$ equal to the maximum value among those leaf sets on instance. Namely, we understand in the worst-case, the precision of BOGPC is $m+1$\\
\end{proof}

~\newline
\textbf{BOERC} is an approach with completely random. First it constructs a roots set $A_q$, afterwards setup each edge relation subgraph by removing the relevant leaves from unit subgraph. Then the unit subgraphs set is converted to an OPERS $\Lambda_q$. While labeling the root, the leaves record the chromatic value of its root. When the leaf is to be labeled as a root, its color would be decided by the record made for it before.\\

~\newline
\renewcommand\arraystretch{1.0}
\begin{longtable}{ p{120mm}}
\toprule  
\textbf{Algorithm 6: BOERC}\\
\toprule  
\textbf{input} graph $G= s_1, s_2,\cdots,s_n;$\\
\qquad\quad   set $A=(x_1, x_2, \cdots, x_n); $\\
\qquad\quad	set $\Omega=(r_1, r_2,\cdots,r_n): r_i=(x_i,~\omega_i=array());$\\
\qquad\quad  Color set $C=(1,~2);$\\
\quad\quad\quad set $R$\\
\textbf{output}  $R$\\
\textbf{00}\quad    \textbf{For }$~1\rightarrow n$ \textbf{ do}\\
\textbf{01}\qquad                 $ L(s_j)\setminus \{x_i\} \Longleftarrow \exists i<j \wedge x_j\in L(\lambda_i) ;$\\
\textbf{02}\quad    \textbf{For }$~1\rightarrow n$ \textbf{ do}\\      
\textbf{03}\qquad      \textbf{If}$(~\Omega(\omega_i )= \varnothing~)$  \textbf{Than }$k=\textbf{random}(C);$\\    
\textbf{04}\qquad 	\textbf{Else If }$(C\setminus \Omega(\omega_i )=\varnothing)$\\     
\textbf{05}\qquad \quad         \textbf{Than }$k=\vert C\vert+1;~C\leftarrow \vert C\vert+1;$ \\
\textbf{06}\qquad\quad   \textbf{Else} $k= \textbf{random}(C\setminus  \Omega(\omega_i )) ;$\\
\textbf{07}\qquad          $R(k)\leftarrow x_i;~ \Omega(\omega_j)\leftarrow k\Longleftarrow \forall x_j\in L(\lambda_i)$\\
\textbf{08}\quad     \textbf{Output }$R=c_1, c_2,\cdots, c_k$(for $c_i=(v_1,v_2,\cdots,v_t$));\\
\bottomrule
\end{longtable}
~\newline
\textbf{Complexity and Algorithm. }In the first loop, the approach converts each unit subgraph to the edge subgraph along the sequence $A_q$. Then it need modify each leaf set, such that the runtime complexity is $O(\vert\tau\vert^2)$. Here, the table of edge subgraph partition likes a unit subgraph for characterizing a directed graph, and there exists at least one empty subgraph like an isolated vertex. But it is truth that this data structure supports the speed-up algorithm. The reason is that the inequivalence-color only need one time to be recorded between the root and leaf. Then the record set increases and the roots and leaves for labeling reduces in the process of coloring. Hence, the runtime complexity of second loop is $O(m^2n^2)$ similarly.%

There is a problem that how is the precision of BOERC. The answer similarly is less than and equal to $m+1$, which would be proved by the following Lemma.

\begin{lemma}\label{cg16}
Let $m$ be the cardinality of a leaf set in unit subgraph. If coloring each vertex in set $V$ on instance $G_{\theta}=(V,~\epsilon)$ by approach BOERC, then there is chromatic value $ K \leq m + 1$.
\end{lemma}

\begin{proof}
Given a no-empty $G_{\theta}=(V,~\epsilon)$. Let $m$ be the cardinality of a leaf set in unit subgraph. There is an approach BOERC for coloring each vertex in set $V$. If $G_{\theta}\in G_{\eta}$, there is $m=n-1$ for each leaf set. Then a vertex at most records $n-1$ chromatic values, such that there is the value $K=n=m+1$ if $G_{\theta}\in G_{\eta}$.%

Consider instance $G_{\theta}\notin G_{\eta}$. Let $m$ be a constant and $\Lambda_q$ be an OPERS on instance. Moreover Let $A_q$ be the roots set of set $\Lambda_q$. There is a number $i$ with $i\geq m$ such that $x_i\in A_q$. If $x_i\in\Lambda_e$, there are at most $m$ records of chromatic values for $x_i$. Then $x_i$ would be labeled with the $(m+1)th$ color. %

If there is a sub-sequence $\alpha\in A_q$ and $\alpha=x_1,x_2,\cdots,x_k$ and $k>i$. we assume that all vertices on sequence $\alpha$ are labeled with $m+1$ colors. When there is a vertex $u$ at the $(k+1)th$ position on sequence $A_q$. If $u\in\Lambda_e$, similarly it has at most $m$ records. Then vertex $u$ would be labeled with the $(m+1)th$ color.%

Summarizing above, we understand the precision of BOERC is $m+1$ in the worst-case. Consider there is the case all cardinalities of leaf sets are difference. In worst-case, we can set the $m$ equals to the maximum cardinality among those leaf sets.\\
\end{proof}

\subsection{Graph Coloring Exp.}
The first object is dodecahedron with $n=20$ and $m=3$. The experimental method is separately implementing the BOGPC and BOERC for 1\,000 times to compare the probability of exploring the minimum chromatic value $K$ and the runtime in practices. These approaches both run on the Apache2.0 server and written by PHP5.0. The reason is the speed of implement is lower than $C++$, so that it is easy to compare the runtime. The results is in following\\

\renewcommand\arraystretch{1.0}
\begin{longtable}{ p{15mm}|| p{20mm}| p{20mm}| p{20mm}| p{24mm}}
\caption{\textbf{EXP. Comparison of BOGPC and BOERC }}\\
\toprule
~&\textbf{K=4 }(1)&\textbf{K=3} (2)&(2)/1000&\textbf{R.T.}\\
\toprule
BOGPC&742&258&0.258&1 s\\ 
\midrule 
BOERC&674&326&0.326&305 ms\\ 
\bottomrule
\end{longtable}
~\\

Whichever options of runtime or probabilistic of exploring the minimum $K$, BOERC is winner. Further we test an instance underlaying two conditions, which numbers of $n,m$ may not change, such that the quantity of edges does not change along with shape changing. The shape of instance can be changed as follow Figure 3.\\
\begin{center}
\includegraphics[height=30mm]{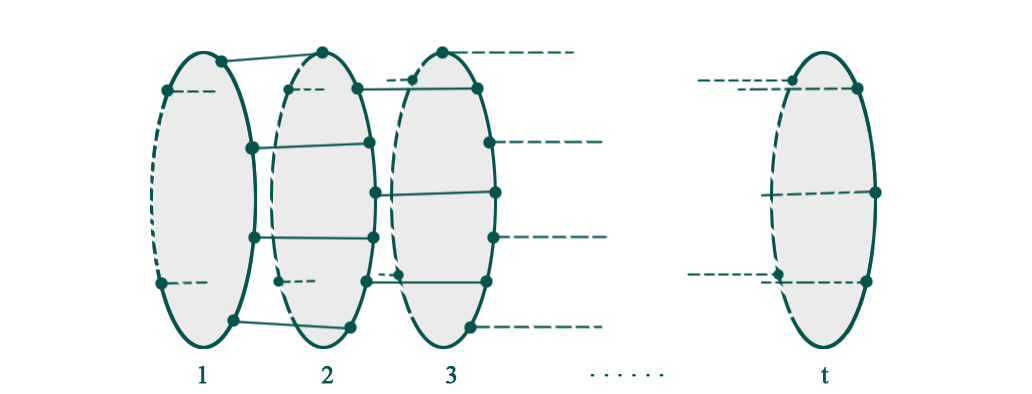}\\
\textbf{Figure 3}
\end{center}

The Figure 3 illustrates the shape: set $s$ vertices on the first and end cycle. There is $2s$ vertices on a medium cycle, such that the $n=2(t-1)s$ and $m=3$. Because of $t=1+ n/2s$, therefore changing the shape can be due to changing the $s$ value. Now we let $n=120$ and $t=3,~4,~5,~6,~7$, thus responsively there is $s=30,~20,~15,~12,~10$.%

Similarly implement two approaches for 1\,000 times to compare the possibilities of exploring chromatic value $K$. We use two Tables to record these results as follow

\renewcommand\arraystretch{1.2}
\begin{longtable}{ p{15mm}|| p{20mm}| p{20mm}| p{20mm}| p{24mm}}
\caption{\textbf{EXP. Testing BOGPC}}\\
\toprule
~&\textbf{s=}&\textbf{K=4} (1)&\textbf{K=3 }(2)&\textbf{R.T.}\\
\toprule
t=3&30&989&11&14 s\\ 
\midrule 
t=4&20&567&433&12 s\\ 
\midrule 
t=5&15&764&236&11.7 s\\ 
\midrule 
t=6&12&735&265&11.4 s\\ 
\midrule 
t=7&10&739&261&11.6 s\\ 
\bottomrule
\end{longtable}

\renewcommand\arraystretch{1.0}
\begin{longtable}{ p{15mm}|| p{20mm}| p{20mm}| p{20mm}| p{24mm}}
\caption{\textbf{EXP. Testing BOERC}}\\
\toprule
~&\textbf{s=}&\textbf{K=4} (1)&\textbf{K=3 }(2)&\textbf{R.T.}\\
\toprule
t=3&30&998&2&2 s\\ 
\midrule 
t=4&20&987&13&2 s\\ 
\midrule 
t=5&15&981&19&2 s\\ 
\midrule 
t=6&12&962&38&2 s\\ 
\midrule 
t=7&10&966&34&2 s\\ 
\bottomrule
\end{longtable}

However the complexity is depended on $mn$, which is the number of $\vert\tau\vert$. And those things deal by both approaches are relation of edges, so that we understand the shape is not major factor to affect the runtimes. The method of BOERC is labeling vertices along the sequence of roots set, such that we get a same value of runtime. But for BOGPC, the shape changing may affect the number of regions in graph partition set $\Sigma$, then there is a little differenc on values of runtime. Consider method of BOERC is constructed by completely random sequence, such that the random scope is larger than ones of  BOGPC. Hence the probabilistic of getting minimum $K$ reduces quickly with the data increasing. We can image that the probabilistic would toward to 0 with a large number. But you can deny the fact that it faithfully give you a precision of $m+1$. %

And the BOGPC is still stable than BOERC, which random scope at most is $n/m$ less than BOERC, but it may be affected by the shape of instance. How to choose two approach is decided by you and your research. However two approach always guaranteed the coloring is correct to satisfy the demand of coloring relation, and their precision are less than and equal to $m+1$.%

There is a remarkable point for grid and completed regular graph. In fact, the instance of $\vert\Delta\vert=2$ in Theorem\ref{cg10} is grid or tree. Yon can understand how to get the $K$ value on grid figure by exploiting method of graph partition. And the completed regular graph, you can obtain its type by checking the data structure.%

~\newline
\textbf{Exact} algorithm for the problem of graph coloring. Certainly, we can enumerate all possibilities in approach of BOGPC like backtrace all possible paths in BOTS. Author tested the object of dodecahedron for this answer. The conclusion is there are plenty of repeated data. As Lemma\ref{cg9_1}, which shows the MCIVS is not an equivalent class in set $\bar{\epsilon}$. Then author give the such theorem as follow

\begin{theorem}\label{cg17}
Let $\Delta$ be a MCIVS on an instance $G_{\theta}=(V,\epsilon)$. Consider set $\Delta\neq\varnothing$ and is a sequence of components like $\Delta =(\delta_1,\delta_2,\cdots,\delta_N)$. Let $i_k=\vert\delta_k\vert$ to represent the cardinality of each component on sequence $\Delta$. Moreover, let $\Lambda_q$ be an OPERS on instance and $A_q$ be roots set of set $\Lambda_q$. If $\Delta$ is on sequence $ A_q$, then there is a number of permutation $p=\prod_{k=1}^N i_k!$ on set $A_q$, such that there is $A_q(p)\leadsto\Lambda_q(p)$. For each subgraph $\lambda_i\in\Lambda_q$, such that there is $\lambda_i(p)=\lambda_i$.
\end{theorem}

\begin{proof}
Given a no-empty graph $G_{\theta}=(V,~\epsilon)$. Let $\Delta$ be a MCIVS on graph $G_{\theta}$ and $\Delta =(\delta_1,\delta_2,\cdots,\delta_N)$, then having $i_k=\vert\delta_k\vert$. Moreover, let $\Lambda_q$ be an OPERS on instance and $A_q$ be roots set of set $\Lambda_q$. We can set sequence $\Delta$ is on sequence $ A_q$ such that $\Delta\subseteq A_q$.%

Consider given each $x_i,x_j\in A_q$.  We can assume there are tow components  $\lambda_i, \lambda_j\in\Lambda_q$ with $x_i\in R(\lambda_i)$ and $ x_j\in R(\lambda_j)$. If pair $x_i,x_j\in \bar{\epsilon}$, we have 
$R(\lambda_i)\cap L(\lambda_j)=\varnothing$ and $R(\lambda_j)\cap L(\lambda_i)=\varnothing$,
as Theorem\ref{cg7}.  Hence, changing positions of vertices $x_i,x_j$ on set $A_q$ can not change $\lambda_i,\lambda_j$. Then if for a component $\delta_k\in\Delta$ and pair $x_i,x_j\in \delta_k$, those subgraph $\lambda_i,\lambda_j$ similarly can be changed by changing positions of pair $x_i,x_j$  in $\delta_k$.%

Consider a vertex $x_t\in\delta_s$ and $\delta_s\neq \delta_k$ with pair $\{x_t, x_i\}\in\epsilon$. Because there is the given condition of $\Delta$ on sequence $ A_q$, therefore changing position of $x_t$ in set $\delta_s$ can not change the values of $k,s$; i.e. if the vertex $x_t$ is at the fore position than vertex $x_i$ on sequence $A_q$, the fact can not be changed by changing position of $x_t$ in set $\delta_s$. It is similar to vertex $x_i$. Hence, there is a number of permutation $i_k!$ of vertices in each component on sequence $\Delta$, i.e. there is a number of permutation $p=\prod_{k=1}^N i_k!$ on set $A_q$, for each component $\lambda_i\in\Lambda_q$ such that $\lambda_i(p)=\lambda_i$.%

Consider a vertex $x_s\in (A_q\setminus\Delta)$. If  there is randomly channging the position of vertex $x_s$ on sequence $A_q$, it is certainly that the case causes each associated edge subgraph change. Hence, $\lambda_i(p)=\lambda_i$ is not a necessary condition in this Theorem. Of cause, it is easy to prove there is no such property on completed regular graph with $\Delta =\varnothing$. Now we prove what condition can lead to edge subgraph partition unchanged in set $\Lambda_q$. We call this Theorem \emph{edge subgraph partition unchanged theorem}.\\
\end{proof}

As the Theorem\ref{cg17}, we have this experiment with dodecahedron as follow. The Table shows the case of which the program enumerated a completed interval set.

\renewcommand\arraystretch{1.0}
\begin{longtable}{ p{30mm}|| p{20mm}| p{20mm}| p{20mm}}
\caption{\textbf{Exp. for Exact Algorithm of Coloring}}\\
\toprule
~&$\vert \Delta\vert=6$&$\vert \Delta\vert=7$&$\vert \Delta\vert=8$\\
\toprule
Original  Data&269&90&2\\ 
\midrule 
Exact Data&6&55&2\\ 
\bottomrule
\end{longtable}

With rough viewpoint, program need do much more works of comparing among arrays to remove the repeated data. Such that the price of runtime and memory increase quickly. Because of $m=3$ on dodecahedron figure, therefore it must at least enumerate CIVSs for two times. Let $x_1(6)=6;~x_1(7)=55;~x_1(8)=2$ represent the first loop, and similarly the second one has $x_2(6)=6;~x_2(7)=55;~x_2(8)=2$. Then there exists a combination formula 
\[X^2=(\sum_{i=1}^2 x_i(k))^2;\qquad k=6,7,8. \]
such that we have a number of combination $(6+55+2)^2=3\,969$, i.e. there is a search breadth of 3\,969. We can evaluate the number of combination for a given instance as follow
\[ X^{m-1}=(\sum_{i=1}^{m-1} x_i(k))^{m-1};\qquad 2\leq k \leq n. \]
We can let $\sum x_i(k)\approx n$, then complexity is $n^3n^{m-1}=O(n^{m+2})$. It is similar to the current approaches.\\

\section{Conclusion}
\textbf{Summary. }In this work we studied how to cut graph actually, with such logic structure $\epsilon\subseteq \tau\subseteq V^2\text{; }~\Lambda\subseteq S \subseteq \tau$. We proved that some problems can be quantified so that it can be a basic relation model for applications. Similarly, we proved some axioms in current theory and show why some things are so hard to us. The algorithmic contribution focused on the data-structure such that solve the problem of general. In the process, the equivalent class, unit subgraph is the keypoint. It let us freely choose the method to abstract basic relation for construct new logic model. For example we abstract the edge relation from it, and finally construct two classes, with those properties of symmetry and transitivity. In fact, there are more methods to abstract this binary relation for problems such as AI, flow network, TSP. Due to limited space of page, author can not continue to do these works.

~\newline
\textbf{Future} Work. A wide range of possible future work exists for present abstract relation among those objects, e.g. TSP. We give the cutting graph by BOGPC or BOERC. For among each CIVS, those edges are the bridges between two arbitrary domains of vertices, and for those domains indeed, there exists the relation of inequivalence-color among them. It lets us may have a nice condition to use the \emph{Greedy Algorithm} to exactly solve this problem, so that graph coloring is not a pure problem of graph theory. Further for AI, we can use graph partition to characterize the process of solving some problems based on $k$ conditions. Finally, we pose the conjecture in following.
\\
~\newline
\textbf{Conjecture. }There are two binary relation $\rho$ and $\bar{\rho}$ with $\rho\neq\bar{\rho}$ on a universal set $U$, which lead to two equivalent classes $\rho(U)$ and $\bar{\rho}(U)$ on $U$ respectively. If for $R=[a_i]_{\rho}\cap [a_i]_{\bar{\rho}}$ such that $R\neq\varnothing$, then $R$ is \emph{Russell Paradox}.\\

\emph{Reason.} Self-cycle appears in graph traversal but vanishing in edge relation. And traversal relation justly possesses property of reflexivity without symmetry, to contrary for edge relation. And we can find the equivalent classes for traversal relation in $V^2$, but not on edge relation.

\end{document}